%% file: 00-Main-cogMI.tex
\documentclass[conference]{IEEEtran}
\IEEEoverridecommandlockouts
\usepackage{comment}
\usepackage[most]{tcolorbox}
\usepackage{multirow}
\usepackage{cite}
\usepackage{amsmath,amssymb,amsfonts}
\usepackage{amsthm}
\usepackage{graphicx}
\usepackage{textcomp}
\usepackage{paralist}
\usepackage{xcolor}
\usepackage[hidelinks]{hyperref}
\usepackage[noend]{algpseudocode}
\usepackage[normalem]{ulem}
\usepackage{subfigure}
\usepackage{tabularx}
\usepackage{array}
\usepackage{color}
\usepackage{tabulary}
\usepackage{colortbl}
\usepackage{comment}
\usepackage{pdflscape}
\usepackage{float}
\usepackage{listings}
\usepackage[]{mdframed}
\usepackage{fvextra}
\usepackage{caption}
\usepackage{paralist}
\usepackage{adjustbox}
\usepackage{amsmath, amssymb}
\usepackage{algorithm}
\usepackage{multicol}
\usepackage{placeins}
\usepackage{algpseudocode}
\usepackage{mathtools}
\usepackage{rotating}
\usepackage{makecell}
\usepackage{seqsplit}
\usepackage[obeyspaces]{xurl}
\usepackage{ragged2e}
\usepackage{booktabs}
\newcolumntype{Y}{>{\RaggedRight\arraybackslash\hspace{0pt}}X}
\newcolumntype{L}[1]{>{\raggedright\arraybackslash}p{#1}}
\newcolumntype{V}[1]{>{\raggedright\arraybackslash}m{#1}}
\newtheorem{defn}{Definition}
\newtheorem{lemma}{Lemma}
\newtheorem{property}{Property}

\newtheorem*{definition*}{Definition}
\newtheorem{prop}{Proposition}
\def\BibTeX{{\rm B\kern-.05em{\sc i\kern-.025em b}\kern-.08em
    T\kern-.1667em\lower.7ex\hbox{E}\kern-.125emX}}
\begin{document}

\title{AWARE-RG: Abstraction Aware Explanation of Resiliency Graphs for Cyber Physical Systems\\
% {\footnotesize \textsuperscript{*}Note: Sub-titles are not captured for https://ieeexplore.ieee.org  and
% should not be used}
% \thanks{Identify applicable funding agency here. If none, delete this.}
}

% \author{\IEEEauthorblockN{1\textsuperscript{st} Given Name Surname}
% \IEEEauthorblockA{\textit{dept. name of organization (of Aff.)} \\
% \textit{name of organization (of Aff.)}\\
% City, Country \\
% email address or ORCID}
% \and
% \IEEEauthorblockN{2\textsuperscript{nd} Given Name Surname}
% \IEEEauthorblockA{\textit{dept. name of organization (of Aff.)} \\
% \textit{name of organization (of Aff.)}\\
% City, Country \\
% email address or ORCID}
% \and
% \IEEEauthorblockN{3\textsuperscript{rd} Given Name Surname}
% \IEEEauthorblockA{\textit{dept. name of organization (of Aff.)} \\
% \textit{name of organization (of Aff.)}\\
% City, Country \\
% email address or ORCID}
% \and
% \IEEEauthorblockN{4\textsuperscript{th} Given Name Surname}
% \IEEEauthorblockA{\textit{dept. name of organization (of Aff.)} \\
% \textit{name of organization (of Aff.)}\\
% City, Country \\
% email address or ORCID}
% \and
% \IEEEauthorblockN{5\textsuperscript{th} Given Name Surname}
% \IEEEauthorblockA{\textit{dept. name of organization (of Aff.)} \\
% \textit{name of organization (of Aff.)}\\
% City, Country \\
% email address or ORCID}
% \and
% \IEEEauthorblockN{6\textsuperscript{th} Given Name Surname}
% \IEEEauthorblockA{\textit{dept. name of organization (of Aff.)} \\
% \textit{name of organization (of Aff.)}\\
% City, Country \\
% email address or ORCID}
% }

\author{
\IEEEauthorblockN{Shadaab Kawnain Bashir, Rakesh Podder, Philip D. Hopkins, Sarath Sreedharan, Indrakshi Ray, and Indrajit Ray}
\IEEEauthorblockA{Computer Science Department, Colorado State University, Fort Collins, Colorado, USA \\ 
\{shadaab.bashir, rakesh.podder, philip.hopkins, sarath.sreedharan,  indrakshi.ray, indrajit.ray\}@colostate.edu}
\\ 
}

\maketitle

\begin{abstract}
\input{01-Abstract}
\end{abstract}


% ------------------------
\section{Introduction} \label{introduction}
\input{10-Introduction-IJR}
% ------------------------
\section{Related Work} \label{related_work}
\input{30-RelatedWork}
\section{Background} \label{background}
\input{20-Background-rakesh}

% ------------------------
\section{AWARE-RG Overview and Security Model}
\label{sec:overview}
\input{Overview}

\subsection{Threat Model}
\input{30-Threat-Model_new}
\label{sec:threat-model}

\subsection{Running Example} \label{example}
\input{55-RunningExample-FlareSystem}

\section{AWARE-RG Lattice for Resiliency Graphs} 
\label{lattice}
\input{60-Lattice_SandP}
%--------------------------------------
\section{AWARE-RG Approach} \label{framework}

\input{41-AWARE_RG_New-SandP}

% ------------------------ 
\section{Evaluation and Analysis} \label{eval}
\input{70-Evaluation}
\input{90-FutureDirectionandConclusion}
\label{future_conclusion}
% % ------------------------
\section*{Acknowledgment}
This work is partially supported by the U.S. Office of
Naval Research (ONR) under award N000142612041. The views and conclusions expressed are those of the authors and do not necessarily reflect those of ONR or other federal agencies.

\bibliographystyle{IEEEtran}
\bibliography{reference}
\input{Appendix_New}
\end{document}

%% file: 01-Abstract.tex
Understanding the security posture of cyber physical systems (CPS) requires reasoning across both information technology (IT) and operational technology (OT). However, a security analyst often holds an incomplete, role-specific view of the underlying system which has most likely resulted from their knowledge limited by a need-to-know framework. Such an incomplete view can potentially lead them to incorrectly assess if a cyber attack produces an untoward physical consequence resulting in making uninformed decisions about system resilience. In this work, we present AWARE-RG (Abstraction Aware Explanation of Resiliency Graphs), a semi-automated, human-in-the-loop framework that provides an analyst with a role-bounded, need-to-know explanations of resiliency paths in a CPS. The AWARE-RG abstraction structure is derived using Formal Concept Analysis over MITRE ATT\&CK for ICS tactics, ordered by the ICS Cyber Kill Chain. For valid attack-fault queries, plan repair and backward means-end analysis identify the missing dependencies, explain why the physical hazard is reachable, and suggest defensive strategies. For invalid attack-fault queries, AWARE-RG localizes the first abstraction level at which the path fails, and reports confirmed blocking conditions as hardening pointers.

%% file: 10-Introduction-IJR.tex
Cyber Physical Systems (CPS), which form the backbone of many critical infrastructures, tightly couple digital control with operational technology (OT) with other information technology (IT) systems. The attack surface of a CPS %is not confined to any single domain; it 
spans IT  and OT assets and their dependencies. %between them.
A compromise in an enterprise IT network can propagate through OT, disrupt programmable logic controllers (PLCs), and ultimately trigger physical hazards \cite{noauthor_blackenergy_2023, 10171308, falliere_w32stuxnet_2011}, potentially causing serious damage.
Managing this attack surface requires collaborative and coordinated reasoning among multiple stakeholders, which includes IT and OT security analysts, engineers, plant managers, and incident responders, with every stakeholder holding a specialized but incomplete view of the system. These incomplete views are often unavoidable owing to {\em need-to-know} and {\em separation of duty} principles in place in corporate structures and as reflected in standards such as IEC 62443 \cite{marszal2019isa} and NIST SP 800-82 \cite{stouffer2023guide}.
%
% The critical security consequence is that stakeholders routinely form mental models of the system that diverge from its actual behavior. 
Thus, the analysts develop a mental model of the system that can diverge from the actual CPS behavior. 
%potentially resulting in improper or inadequate security solutions.

Such incomplete visibility into CPS security can result in erroneous decisions. An analyst might have a mental model of an attack scenario (referred to as an attack path) succeeding when in reality it does not, or vice versa, misallocating defensive resources toward an unknown threat. Alternatively, an analyst who does not understand why a path fails cannot reliably maintain the conditions required to block the propagation of an attack. If these knowledge gaps are left unaddressed, they can weaken even well resourced security measures. Security analysts need to understand how an attack succeeds or why it fails, but that understanding must be calibrated to their knowledge level and bounded by their organizational role. Existing frameworks on CPS security analysis, such as those in \cite{bashir2024, akbarzadeh2022unified, shilenge2022optimization, ray2023explainable}, fail
to address this problem.

In this work, we present AWARE-RG (Abstraction Aware Explanation of Resiliency Graphs), a semi-automated framework that balances two competing needs: providing IT and OT analysts with sufficient cross-domain information to correct gaps in their understanding, while limiting that information to what is necessary and authorized under the \textit{need-to-know} principle. AWARE-RG adopts the Resiliency Graph (RG) framework of Bashir et al.~\cite{bashir2024}, which uses AI planning to integrate attack graphs from the IT side of a CPS with fault trees from the OT side, associating vulnerabilities with system failures so that operators can evaluate and manage system resiliency. We chose RG for two reasons:
\begin{inparaenum}[(i)]
    \item it models cross-domain dependencies between IT vulnerabilities and OT safety deterministically, which is our problem space, and
    \item AI planning aligns well with the causal formalism of human reasoning, which we use to generate human understandable explanations.
\end{inparaenum}
Since an RG captures the complete IT and OT security knowledge, sharing the full model with every analyst would violate \textit{need-to-know} constraints. AWARE-RG therefore provides each analyst with an explainable abstraction of the RG that is appropriate to their authorized role and expertise.

AWARE-RG builds on the Hierarchical Expertise-Level Modeling (HELM) paradigm~\cite{sreedharan2018hierarchical}, which organizes domain knowledge into a lattice of abstractions representing different levels of expertise. For a real-world CPS, the lattice generated by HELM is very large; we reduce it by grouping related attack indicators and impact conditions into semantically coherent clusters using Formal Concept Analysis (FCA)~\cite{wille1982restructuring} guided by MITRE ATT\&CK for ICS tactics~\cite{noauthor_mitre_ics}, and ordering those clusters by the ICS Cyber Kill Chain~\cite{assante2015industrial} to reflect attack progression. The resulting smaller lattice preserves RG's semantics and supports expertise relevant explanations.

% AWARE-RG is derived from
% %This is accomplished by using 
% the Hierarchical Expertise-Level Modeling (HELM) paradigm \cite{sreedharan2018hierarchical}, which presents knowledge at different levels as a lattice of abstractions.
% For a real-world CPS, the lattice generated by HELM is very large. 
% We reduce it meaningfully by grouping related attack indicators and impact conditions into semantically coherent clusters using Formal Concept Analysis (FCA) \cite{wille1982restructuring}, guided by MITRE ATT\&CK for ICS tactics  \cite{noauthor_mitre_ics}, and ordering the clusters via the ICS Cyber Kill Chain \cite{assante2015industrial} model to reflect natural attack progression. This yields a smaller lattice without loss of semantic fidelity of RG, and is used to generate expertise relevant and analyst visibility specific explanations.

Our key contributions are as follows:
\begin{inparaenum}[(i)]
    \item We propose AWARE-RG, a semi-automated framework that augments RG with security analyst knowledge aware explanation. %for both valid and invalid attack-fault paths in CPS.
    \item We introduce a systematic way to create an abstraction lattice, which we term as AWARE-RG lattice for CPS attack-fault path reasoning using MITRE ATT\&CK for ICS and the ICS Cyber Kill Chain, capturing stages of adversarial behavior propagation. 
    \item We demonstrate how the AWARE-RG lattice can be used in CPS security, 
    % allowing identification of matches and mismatches between stakeholder understanding and actual system behavior 
    % \textcol
    localizing the abstraction boundary at which an analyst's mental model diverges from the actual system behavior.
    \item Implement AWARE-RG as a human-in-the-loop pipeline operating on planning compilations, and instantiate it on an LNG flare system CPS. %the Plan Fragment Mapper, which maps user steps to plan fragments; the Path Validity Analyzer, which analyzes why a path is executable; the Path Invalidity Analyzer, which analyzes why a path is not executable; and the Explanation Generator, which identifies the cause of success or failure and produces a minimal, stakeholder knowledge aligned natural language explanation.
    \item Finally, evaluate AWARE-RG using case studies over an industrial scenario. %demonstrating that valid paths explain feasible attack-fault propagation while invalid paths expose the system's resilience boundaries.
\end{inparaenum}

The rest of the paper is organized as follows: \S\ref{related_work} summarizes relevant literature;  
% Section () provides a brief overview of how AWARE-RG works.
\S\ref{background} provides the background of the study; 
% relevant for the study. 
\S\ref{sec:overview} gives an overview of AWARE-RG,
defines the threat model, 
and provides a running example of a CPS that is used to explain the rest of our work;  
\S\ref{lattice} describes the AWARE-RG lattice construction;
\S\ref{framework} describes the rest of the steps in the AWARE-RG pipeline; 
\S\ref{eval} evaluates our framework.
% Discussion of our study is provided in Section \ref{discussion}. 
Finally, \S\ref{future_conclusion} concludes the paper.
%\textcolor{blue}{please check some sections are missing}\textcolor{red}{Fixed}

%% file: 30-RelatedWork.tex
Attack graphs have been used to model how an attacker can chain vulnerabilities and privilege gains to compromise a system, for defense prioritization, and security analysis \cite{jha2002two, gonda2018analysis, 5936075,rakesh2025spear,podder2026exam,talukder2025vkg2ag}. 
These approaches assume that the analysts posing the queries are allowed to get the complete picture of the attack. 
In complex CPS, this assumption does not hold. In practice, roles are given information on a \textit{need-to-know} basis.
%Overall, these approaches are well suited for questions such as which paths an attacker might take to reach a target, which vulnerabilities to patch first to reduce system risk, etc. 
%However, they do not address interactive debugging of a user proposed attack scenarios or plan to reach the outcome, nor do they provide abstraction aware minimal explanations of why such scenario succeeds or fails, i.e., what crucial information is missing relative to the user’s knowledge abstraction.
%
Provenance graphs have been used in cybersecurity to explain why an outcome is chosen and to support the evidence behind it \cite{friedman2020provenance,wu2025provx}.
%To support plan evaluation and inspection Friedman et al. attaches provenance to plans \cite{friedman2020provenance}. Wu et al. proposed PorvX, a framework that explains black box Intrusion Detection System (IDS) predictions over provenance graphs by generating counterfactual explanations for Graph Neural Network (GNN) based Advanced Persistent Threat (APT) detection by identifying minimal subgraph changes that flip predictions \cite{wu2025provx}. 
Provenance graphs mainly provide answer to the question about why a prediction is made rather than explaining why a user proposed attack scenario fails across domains.

Planning based methodologies have been used to generate attack strategies and perform ``what-if'' analysis, while contrastive approaches explain why a planner chose one course of action over a user's alternative. Katz et al. justify planned actions through causal chains with a Graphical user interface (GUI) to limit explanation overload \cite{katz2017autonomous}. Miller compiles user foils into model restrictions and highlights plan differences \cite{miller2021contrastive}, Kim et al. mine contrastive linear temporal logic (LTL) specifications distinguishing two sets of plan traces \cite{kim2019bayesian} and Sukkerd et al.'s give verbal tradeoff explanations for multiobjective probabilistic planning \cite{sukkerd2018multiobjective}. In parallel, XAI work in CPS explains machine learning anomaly and intrusion detectors in terms of input features \cite{vigano2020explainable,farhan2025explainable,bose2025hybrid, kumar2025xai}, and Hoenig et al. survey the resulting challenges around context and user tailored explanation \cite{hoenig2024explainable}. These lines of work explain a model's own output, i.e., a planner's choice or a detector's prediction, rather than debugging an analyst's proposed attack scenario. None localizes where an analyst's mental model diverges from the actual system, and none bounds disclosure by organizational role, which is the gap AWARE-RG addresses.

%% file: 20-Background-rakesh.tex
\textbf{Resiliency Graph (RG)} models how cyber attacks and system failures interact in a CPS~\cite{bashir2024}. An RG is a directed graph $RG = \langle N, E_A, E_F \rangle$, where $N = H \times 2^{\mathbb{P}}$ is the set of nodes, with $H$ the set of hosts and $\mathbb{P}$ the set of propositional host attributes; each node $n = (h,\mu)$ represents a host $h$ together with its current attribute state $\mu \subseteq \mathbb{P}$. $E_A \subseteq N \times N$ is the set of \emph{attack edges}, each representing a potential attack step, and $E_F \subseteq N \times N$ is the set of \emph{fault edges}, each representing a potential fault propagation step. An edge $e = (n_i, n_j) \in E_A \cup E_F$ with $n_i = (h_i, \mu_i)$ and $n_j = (h_j, \mu_j)$ is enabled only if its preconditions hold at the source, $\mathit{pre}(e) \subseteq \mu_i$, and its
postconditions hold at the target, $\mathit{post}^{+}(e) \subseteq \mu_j$ and $\mu_j \cap \mathit{post}^{-}(e) = \emptyset$.

\smallskip
\noindent\textbf{Classical Planning}
A classical planning model is a tuple $M = \langle F, A, I, G \rangle$, where $F$ is a set of propositional fluents, $A$ a set of actions, $I \subseteq F$ is the initial state, and $G \subseteq F$ is the goal. A state is a set $s \subseteq F$ of fluents that hold. Each action $a \in A$ is a triple $a = \langle \mathit{pre}(a), \mathit{add}(a), \mathit{del}(a) \rangle$, \emph{applicable} in state $s$ if $\mathit{pre}(a) \subseteq s$, yielding $s' = (s \setminus \mathit{del}(a)) \cup \mathit{add}(a)$. A \emph{plan} is a sequence $\pi = \langle a_1, \dots, a_k \rangle$; we write $\mathit{Valid}(M, \pi) = \top$ and call $\pi$ \emph{valid} in $M$ if each $a_i$ is applicable in the state resulting from $\langle a_1, \dots, a_{i-1} \rangle$ applied to $I$, and the final state satisfies $G \subseteq s_k$. Unless stated otherwise, \emph{predicate} refers throughout to a grounded predicate, that is, a fluent obtained by binding a predicate schema's variables to domain objects. Since every action carries explicit preconditions and effects, planning traces identify precisely where a proposed action sequence succeeds or fails, making the formalism suited to machine checkable validation and \textit{need-to-know} explanation.

\noindent\textbf{Projection Abstraction}
Given a planning model $M = \langle F, A, I, G \rangle$ and a fluent set $P \subseteq F$, the \emph{projection} of $M$ with respect to $P$ is $f_P(M) = \langle F \setminus P,\, A|_P,\, I \setminus P,\, G \setminus P \rangle$, where $A|_P$ contains, for each $a \in A$, the action $\langle \mathit{pre}(a) \setminus P,\, \mathit{add}(a) \setminus P,\, \mathit{del}(a) \setminus P \rangle$. Projection induces the state abstraction $\alpha_P(s) = s \setminus P$, mapping each transition $\langle s_1, a, s_2 \rangle$ to $\langle \alpha_P(s_1), a, \alpha_P(s_2) \rangle$, following Sreedharan et al.~\cite{sreedharan2018hierarchical}.

\smallskip
\noindent\textbf{Hierarchical Expertise-Level Modeling (HELM)}
generates user specific explanations for individuals with varying levels of domain understanding ~\cite{sreedharan2018hierarchical}. It represents the user's knowledge as an abstraction
of the explainer's complete domain model and organizes candidate abstractions into a model lattice. A \emph{foil} is the scenario the user believes will produce a particular outcome. Given a foil, HELM
identifies the user's mental model as the most abstract model in which the foil still appears valid, then generates a minimal explanation as the smallest set of concretizations required to invalidate the foil.
% , computed via $A^{*}$ search or a greedy set cover procedure.

%% file: Overview.tex
\begin{figure*}[!h]
    \centering
    \includegraphics[width=.7\linewidth]{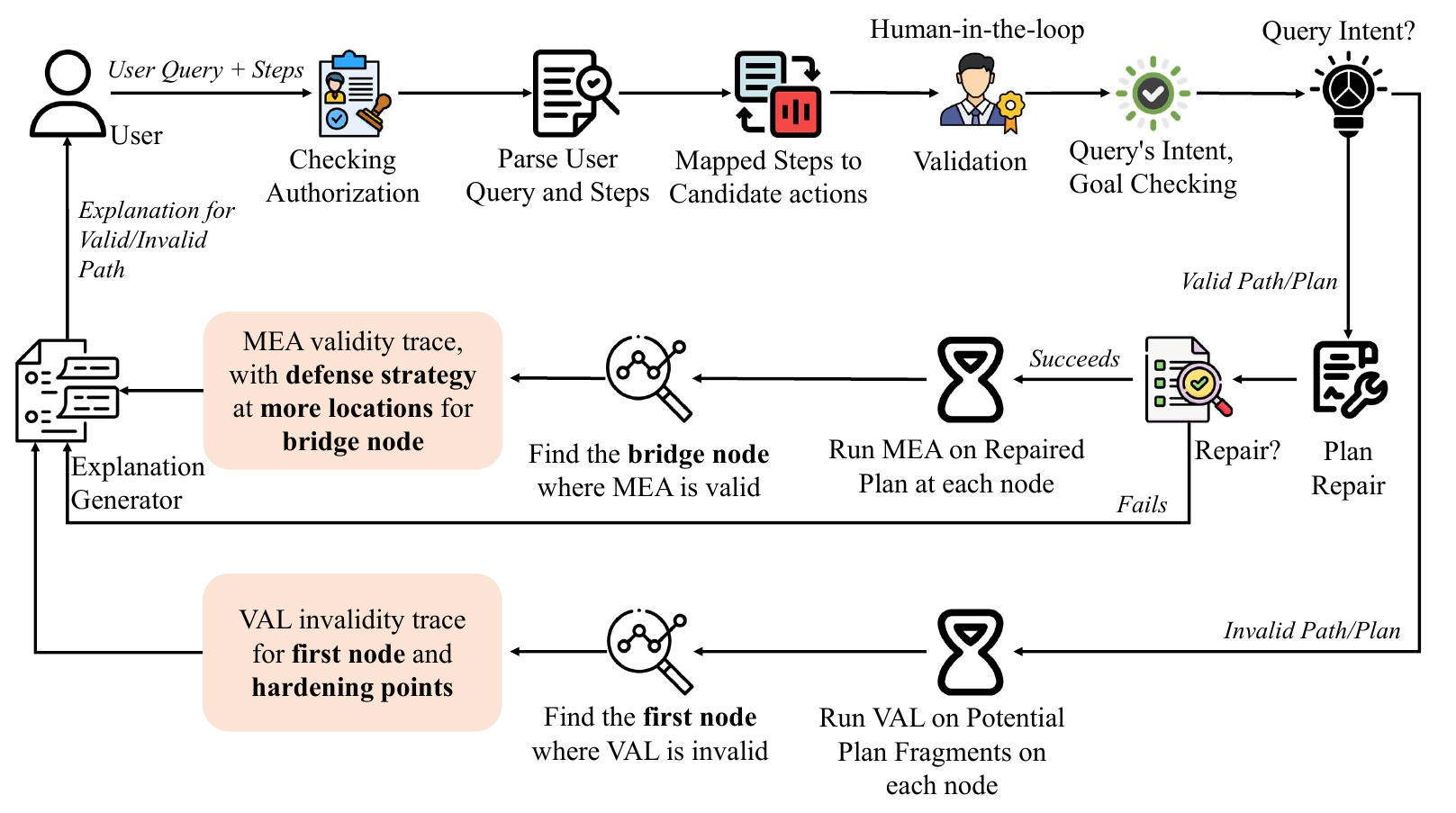}
    \caption{AWARE-RG explanation pipeline with guided abstraction search.}
    \label{fig:AWARE-RG_pipeline}
\end{figure*}

This section introduces AWARE-RG, followed by the threat model and the running example used throughout the paper.

\subsection{Overview of AWARE-RG}
AWARE-RG is an analyst aware security analysis framework for explaining whether a proposed attack-fault path can propagate through a CPS and trigger physical consequences, addressing the problem raised in \S\ref{introduction}, where IT and OT analysts reason about such paths using incomplete and role specific knowledge. It resolves this problem by validating analyst aware attack-fault paths against the RG's planning model while enforcing a \textit{need-to-know} disclosure policy, where each analyst receives only the abstraction of the RG model required to reason about their query, bounded by their authenticated role and authorized abstraction view.

AWARE-RG achieves this by organizing the CPS's RG into IT and OT analysts' partial views as a lattice structure, called the AWARE-RG lattice (\S\ref{lattice}), created by grouping RG's predicates into MITRE ICS ATT\&CK tactic clusters and determining an analyst's position in this lattice from their queries. The analyst submits a natural language query describing an expected attack-fault path, along with why they believe it does or does not lead to physical consequences. AWARE-RG translates the provided steps into grounded actions using the planning compilation of RG and infers the query intent, that is, whether the analyst asks why an attack-fault path exists or why it does not. Existence queries are handled by Path Validity Analysis (\S\ref{valid_plan}), which repairs incomplete paths and applies backward MEA to identify the bridge node containing sufficient information to explain why the hazard occurs; the repaired conditions later form the sole basis for the defensive measures suggested to the analyst. Non-existence queries are handled by Path Invalidity Analysis (\S\ref{invalid_plan}), which localizes the first abstraction node at which the path fails and extracts the unsatisfied preconditions as hardening pointers, which ground the recommendations returned to the analyst. The Explanation Generator (\S\ref{explanation_generator}) converts the resulting traces into \textit{need-to-know} explanations, confining every claim to the predicates visible at the analyst's authorized view. AWARE-RG thereby supports security decision making by validating analyst assumptions, localizing incomplete attack-path reasoning, identifying resilience boundaries, and producing defensive guidance without exposing the complete RG. Figure~\ref{fig:AWARE-RG_pipeline} illustrates the pipeline, which is described in detail in \S\ref{framework}. Domain experts construct the RG, validate tactic assignments, define need-to-know boundaries, and review step mappings, after which attack-fault path analysis and explanation generation run automatically.

%% file: 30-Threat-Model_new.tex
\textbf{System and trust assumptions} AWARE-RG is deployed as a cross-domain security analysis tool in a CPS roughly following the layered Purdue Reference Architecture \cite{williams1994purdue}, where layer specific experts, such as IT and OT analysts, submit queries to reason about attack-fault paths and receive explanations. We assume RG is a trusted AI-planning based model of the CPS built from known assets, configurations, dependencies, and vulnerabilities. The threat model focuses on an authorized but malicious analyst who queries AWARE-RG to obtain information beyond their \textit{need-to-know} boundary.

\textbf{Adversary model} We consider an adaptive adversary with a view of the CPS commensurate with some Purdue layer, who knows the AWARE-RG architecture (Kerckhoffs's principle \cite{petitcolas2025kerckhoffs,kerckhoff1983nielsen}) but not the content of the AWARE-RG lattice, and who interacts only through the query interface, issuing an adaptively chosen sequence of queries and observing all returned outputs. 
We consider three types of attacks.
% \begin{enumerate}

\smallskip
\noindent \textbf{Semantic mapping attack ($T_1$):} The adversary crafts a query whose steps are syntactically valid but semantically incorrect, causing the mapper to ground incorrect actions and producing an explanation that misdirects the expert's defensive response.

\smallskip
\noindent \textbf{Repair probing attack ($T_2$):} By submitting deliberately incomplete queries and observing which missing actions are returned, the adversary attempts to reconstruct the hidden RG model's
dependency structure, a confidentiality threat.

\smallskip
\noindent \textbf{Explanation inference attack ($T_3$):} By comparing which predicates appear in validity traces, repair reports, and hardening pointers across queries, the adversary determines the predicate level abstraction boundaries of other Purdue layer experts, identifying which
attack stages will not be flagged at adjacent layers.
% \end{enumerate}
Information disclosure via side channel analysis is outside the scope of this work. 

%% file: 55-RunningExample-FlareSystem.tex
We use the LNG at GL1/Z-SONATRACH, Algeria \cite{berrouane2016safety,kabir2023dynamic}, comprising tank, cold, and hot flare systems, whose procedures are automated by several PLCs to ensure the flare runs effectively, safely, and in accordance with standards. A Cisco Adaptive Security Appliance firewall and a Microsoft Windows Server connect these PLCs to the remote desktop. An attacker can reach the internal network either by exploiting the firewall (CVE-2018-0296) or via remote code execution vulnerabilities on the Windows Server (CVE-2019-0575, CVE-2019-0584, CVE-2019-0538), from which OT safety events can be triggered. The schematic and the vulnerabilities linked to each PLC are given in Figure~\ref{fig:flare} and Table~\ref{tab:plc_cves}, Appendix \ref{app:flare}. 

We recreate the flare system's RG model from the artifacts\footnote{GitHub:\url{https://github.com/skawnain04/Resiliency_Graph}} of \cite{bashir2024}, capturing the cross-layer dependencies between IT and OT, and compile it into a classical planning model.

%% file: 60-Lattice_SandP.tex
In CPS security, \textit{separation-of-duty} and \textit{need-to-know} principles give IT and OT personnel partitioned views of the same system. AWARE-RG represents each analyst's knowledge as a partial view ordered by detail level, so that a query can be evaluated at the knowledge level the analyst actually holds and explanations can be generated without disclosing unauthorized information. HELM provides this representation: each node is an admissible projection of $M^{RG}$ retaining a subset of concepts, and the ordering between nodes captures how knowledge states relate across the IT/OT boundary.

\begin{figure}[!h]
    \centering
    \includegraphics[width=\linewidth]{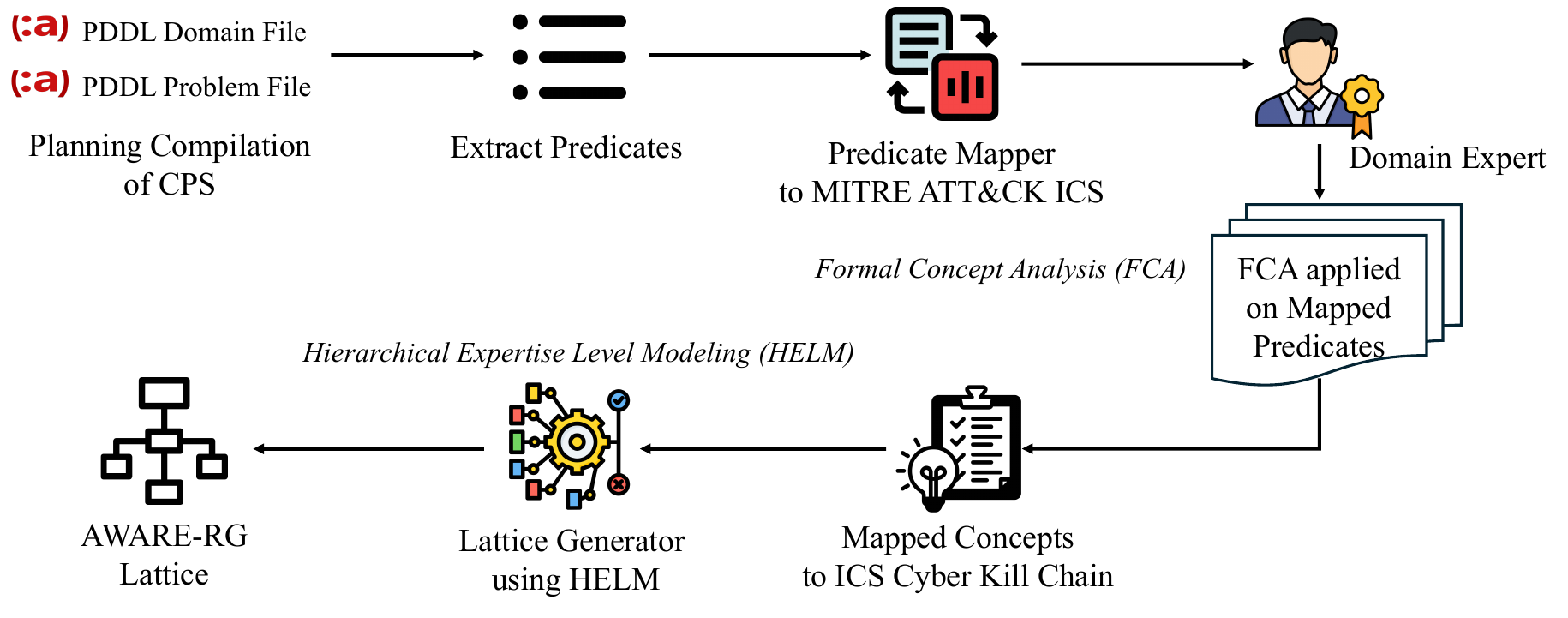}
    \caption{An overview of the Lattice creation pipeline.}
    \label{fig:lattice_pipeline}
\end{figure}

We adapt HELM with three security modifications, shown in \autoref{fig:lattice_pipeline}. We use FCA over MITRE ATT\&CK for ICS tactic assignments to derive security oriented knowledge units, then order them by ICS Cyber Kill Chain stage to preserve the causal progression of cross-domain attack-fault paths. The resulting nodes remain admissible projections of $M^{RG}$; whose soundness is formalized in Lemma~\ref{lem:soundness}.

RG represents cyber attacks and physical faults as a planning model \cite{bashir2024}, from which we define path status:

\begin{defn} \label{def:path-status}
    Let $M^{RG} = \langle F^{RG}, A^{RG}, I^{RG}, G^{RG} \rangle$ be the planning compilation of an RG and $\Pi = \langle a_1, \dots, a_k \rangle$ a grounded action sequence representing an analyst's proposed attack-fault path. $\Pi$ is \emph{valid} in $M^{RG}$ if $\mathit{Valid}(M^{RG}, \Pi) = \top$. $\Pi$ is \emph{repairable} in $M^{RG}$ if it is not valid but there exists a sequence $\Pi'$, obtained from $\Pi$ by inserting actions, with $\mathit{Valid}(M^{RG}, \Pi') = \top$; the preconditions established by the inserted actions constitute the analyst's \emph{missing knowledge}. $\Pi$ is \emph{blocked} in $M^{RG}$ if no such $\Pi'$ exists; the unsatisfied preconditions of the earliest failing action are the \emph{blocking conditions} and serve as hardening pointers.
\end{defn}

AWARE-RG represents analyst knowledge levels via the projection based abstraction of HELM~\cite{sreedharan2018hierarchical}, restated below.

\begin{defn}
\label{model_abs}
%\textbf{Abstract Model}\\
For a planning model $M = \langle F, A, I, G\rangle$ with a corresponding transition system $\mathcal{T}$, a model $M' = \langle F', A', I', G'\rangle$ with a transition system $\mathcal{T}'$ is considered an \textbf{abstraction of $\pmb{M}$}, if there exists a propositions set $P$, such that $F' = F \setminus P$, $I' = I\setminus P$, $G' = G\setminus P$ and for every transition $s_1 \xrightarrow{a} s_2$ in $\mathcal{T}$ corresponding to an action $a$, there exists an equivalent transition  $(s_1 \setminus P) \xrightarrow{a} (s_2 \setminus P)$ in $\mathcal{T}'$. 
\end{defn}
Abstraction induces a partial order over models: $M_1 \preceq M_2$ iff $M_2 = f_P(M_1)$ for some $P \subseteq F$; models higher in the order are more abstract. In AWARE-RG, fluents are not projected out individually but in semantically coherent groups. %Given a partition $\mathcal{K} = \{K_1, \dots, K_m\}$ of $F$, a projection $f_Q$ is \emph{admissible} with respect to $\mathcal{K}$ if $Q$ is a union of blocks of $\mathcal{K}$. The admissible projections of $M$, ordered by $\preceq$, form a bounded-below partial order $\langle \mathcal{M}, \preceq \rangle$ with $M$ as its unique most concrete element; 
Given a partition $\mathcal{K} = \{K_1, \dots, K_m\}$ of $F$, a projection $f_Q$, $\mathcal{M} = \{ f_Q(M) \mid Q \text{ is a union of blocks of } \mathcal{K} \}$ denote the set of all \emph{admissible} projections of $M$. The admissible projections, ordered by $\preceq$, form a
bounded below partial order $\langle \mathcal{M}, \preceq \rangle$ with $M$ as its unique most concrete element (since $f_\emptyset(M) = M$);
we refer to this structure as an \emph{abstraction structure}. The inverse of projection is \emph{concretization}: for an abstract model $M' \in \mathcal{M}$ and $Q \subseteq P$, the concretization $f_Q^{-1}(M')$ satisfies $f_Q(f_Q^{-1}(M')) = M'$, refining $M'$ by reintroducing the fluents in $Q$.

In AWARE-RG, Definition~\ref{model_abs} is instantiated with $M = M^{RG} = \langle F^{RG}, A^{RG}, I^{RG}, G^{RG} \rangle$; all abstraction chain arguments in \S\ref{lattice}--\ref{framework} refer to this instantiation.
We reiterate the following lemma from  Sreedharan et al. \cite{sreedharan2018hierarchical}.
\begin{lemma}[Abstraction soundness]
\label{lem:soundness}
Let $M = \langle F, A, I, G \rangle$ be a planning model and $P \subseteq F$. Every plan $\pi$ valid in $M$ is valid in $f_P(M)$.
\end{lemma}

\begin{prop}[Failure monotonicity]
\label{prop:monotonicity}
If $\Pi$ is invalid in $f_P(M)$, then $\Pi$ is invalid in $M$ and in every $f_Q(M)$ with $Q \subseteq P$. Hence, along any chain ordered from most abstract to most concrete, the set of nodes at which $\Pi$ is invalid is downward closed and the first invalid node is well-defined.
\end{prop}

By contraposition of Lemma~\ref{lem:soundness}, validity in $f_Q(M)$ with $Q \subseteq P$ would imply validity in $f_P(M)$, since $f_P(M) = f_{P \setminus Q}(f_Q(M))$ is a projection of $f_Q(M)$.

\input{62-Predicate-Mapper_NEW}

\input{63-FCA_NEW}
\input{64-Concepts_to_SAN}
\input{65-Lattice-Generation_NEW}

%% file: 62-Predicate-Mapper_NEW.tex
% \subsection{Predicate Mapper}
% \label{sec:predicate-mapper}
\smallskip
\noindent \textbf{Predicate Mapper} \label{sec:predicate-mapper} associates each predicate in the fluent set $F^{RG}$ with MITRE ATT\&CK for ICS tactics, giving each predicate a security interpretation from which subsequent stages derive attack progression oriented knowledge units. Given a concise natural language summary of the CPS operational context covering asset inventory, network topology, and the cyber-to-physical fault propagation, the predicate set $P \subseteq F^{RG}$, and the fixed tactic set $\mathbb{T}$, an LLM produces a tactic assignment $\tau : P \to 2^{\mathbb{T}}$, mapping each predicate to a tactic with justification. Predicates sharing a tactic assignment are grouped into a single knowledge unit. Rather than consuming a single LLM pass, the domain expert samples multiple outputs and selects the assignment that best reflects each predicate's causal role in the attack-fault chain rather than its syntactic similarity to a tactic name, preventing a hallucinated assignment from placing a knowledge unit at the wrong propagation stage. Table~\ref{tab:predicate-fca} shows the mappings for the running example.

%% file: 63-FCA_NEW.tex
\smallskip
\noindent \textbf{Formal Concept Analysis (FCA)} \label{fca} groups the mapped predicates into coherent knowledge units. For AWARE-RG, FCA operates over a formal context $\kappa = (P,\, \mathbb{T}, \varphi)$, where the predicate set $P \subseteq F^{RG}$ is the set of objects, the tactic set $\mathbb{T}$ the set of attributes, and the incidence relation $\varphi \subseteq P \times \mathbb{T}$ is defined by $(p, t) \in \varphi$ iff $t \in \tau(p)$. Predicates sharing the same tactic assignment form a single \emph{knowledge unit} and are revealed or removed together as the abstraction structure transitions between levels.
This reflects how analyst knowledge is organized in practice: an IT analyst's boundary is better represented by whether they understand the Lateral Movement stage than by whether they know each CVE predicate within it, so knowledge units capture analyst knowledge at the level of attack propagation concepts rather than isolated predicates.

% FCA groups the mapped predicates into coherent knowledge units. For AWARE-RG, FCA operates over a formal context $\kappa = (P,\, \mathbb{T}, \varphi)$, where the predicate set $P \subseteq F^{RG}$ is the set of objects, the tactic set $\mathbb{T}$ is the set of attributes, and the incidence relation $\varphi \subseteq P \times \mathbb{T}$ is defined by $(p, t) \in \varphi$ iff $ t \in \tau(p)$, i.e., predicate $p$ is incident to tactic $t$ if and only if $t$ appears in the tactic assignment produced by the Predicate Mapper. Predicates that share the same tactic assignment are treated as a single \emph{knowledge unit} and are revealed or removed together as the abstraction structure transitions between levels. This grouping reflects how a security analyst's knowledge is organized in practice: an IT analyst's knowledge boundary is better represented by whether they understand the Lateral Movement stage than by whether they know each individual CVE predicate within that stage. The FCA derived knowledge units therefore capture analyst knowledge at the level of attack propagation concepts rather than isolated predicates.

These knowledge units define the partition $\mathcal{K} = \{K_1, \dots, K_m\}$ of $F^{RG}$:
% introduced in \S\ref{background}: 
each block $K_i$ is the set of predicates sharing a common tactic assignment, and all admissible projections in the AWARE-RG abstraction structure remove unions of these blocks. Hence no abstraction node can expose a semantically incoherent predicate set, such as a state in which one CVE's exploitability is visible while another CVE serving the same role in attack propagation is removed. In Table~\ref{tab:predicate-fca}, \path{vul-CVE-2019-0575} and \path{vul-CVE-2019-0584} share \textit{Execution} as their tactic, forming knowledge unit $K_3$, and are revealed or removed together at every abstraction.

\begin{table}[t]
    \centering
    \caption{Predicate mappings and FCA knowledge units in the running example. Predicates with identical tactic assignments belong to the same knowledge unit.}
    \label{tab:predicate-fca}
    \scriptsize
    \setlength{\tabcolsep}{3pt}
    \begin{tabularx}{\columnwidth}{
        >{\raggedright\arraybackslash}X
        >{\raggedright\arraybackslash}p{0.37\columnwidth}
        % >{\raggedright\arraybackslash}p{0.17\columnwidth}
        c}
        \toprule
        \textbf{Predicate} &
        \textbf{Tactic} &
        % \textbf{Secondary tactic} &
        \textbf{Unit} \\
        \midrule

        \path{compromised-customer-pc}
        & Initial Access
        % & --
        & $K_1$ \\

        \path{access-to-windows-server}
        & Lateral Movement
        % & --
        & $K_2$ \\

        \path{vul-CVE-2019-0575}
        & Execution
        % & Discovery
        & $K_3$ \\

        \path{vul-CVE-2019-0584}
        & Execution
        % & --
        & $K_3$ \\

        \bottomrule
    \end{tabularx}
\end{table}

%% file: 64-Concepts_to_SAN.tex
% \subsection{ICS Cyber Kill Chain Mapper}\label{sans_mapper}

\smallskip
\noindent \textbf{ICS Cyber Kill Chain Mapper} \label{sans_mapper}
organizes FCA derived knowledge units according to the progression of attacks from IT to OT. \textit{Stage 1} covers IT activities such as intrusion, lateral movement, and ICS attack preparation; \textit{Stage 2} covers OT activities such as ICS execution, process control impairment, and physical consequence, aligning the causal structure of cross-domain attacks with the organizational knowledge boundaries.

Table~\ref{tab:killchain} provides the mapping of MITRE ATT\&CK for ICS tactics to ICS Cyber Kill Chain stages ordered $S_1$ to $S_{12}$. Each knowledge unit $K_i \in \mathcal{K}$ is assigned a stage index $S_i$, inducing a total order over $\mathcal{K}$ that governs which blocks are removed first as abstraction increases. It produces two base nodes: an IT base node retaining Stage 1 concepts and an OT base node retaining Stage 2 concepts, which serve as the roots of the two lattice branches. 
% Within each branch, knowledge units are removed in reverse ICS Cyber Kill Chain order as abstraction increases, .
As abstraction increases, the IT branch removes knowledge units in reverse kill chain order, while the OT branch removes the earliest stages first, so that every node preserves the causal progression of the attack-fault path.

\begin{table}[t]
\centering
\caption{Two stage mapping of MITRE ATT\&CK for ICS tactics to the ICS Cyber Kill Chain. Rows marked $\star$ are instantiated in the running example. \textbf{SEDE}: Sustainment, Entrenchment, Development \&
Execution.}
\label{tab:killchain}
\scriptsize
\setlength{\tabcolsep}{1.5pt}
\resizebox{\columnwidth}{!}{%
\begin{tabularx}{\columnwidth}{
    @{}c|
    >{\raggedright\arraybackslash}p{0.3\columnwidth}
    >{\raggedright\arraybackslash}X
    cc@{}}
\toprule
\textbf{Stage} & \textbf{MITRE ATT\&CK for ICS Tactic}
& \textbf{ICS Cyber Kill Chain Stage} & \textbf{Order} & \textbf{Ex.} \\
\midrule
\multirow{9}{*}{\rotatebox{90}{Stage 1 (IT)}}
& Initial Access        & Cyber Intrusion & $S_1$ & $\star$ \\
& Execution             & Cyber Intrusion & $S_2$ & $\star$ \\
& Persistence           & SEDE & $S_3$ & \\
& Privilege Escalation  & SEDE & $S_4$ & \\
& Command and Control   & Management \& Enablement & $S_5$ & \\
& Discovery             & SEDE & $S_6$ & $\star$ \\
& Lateral Movement      & SEDE & $S_7$ & $\star$ \\
& Collection            & SEDE & $S_8$ & \\
& Evasion               & SEDE & $S_9$ & \\
\midrule
\multirow{5}{*}{\rotatebox{90}{Stage 2 (OT)}}
& Execution                 & ICS Attack & $S_2$ & $\star$ \\
& Lateral Movement          & ICS Attack & $S_7$ & $\star$ \\
& Inhibit Response Function & ICS Attack & $S_{10}$ &  \\
& Impair Process Control    & ICS Attack & $S_{11}$ & $\star$ \\
& Impact & Physical / Operational Consequence & $S_{12}$ & $\star$ \\
\bottomrule
\end{tabularx}}
\end{table}

%% file: 65-Lattice-Generation_NEW.tex
% \subsection{Lattice Generator}
% \label{lattice_generator}

\smallskip
\noindent \textbf{Lattice Generator} \label{lattice_generator}
organizes the ICS Cyber Kill Chain knowledge units into the AWARE-RG lattice using HELM. Each node corresponds to an admissible projection $f_Q(M^{RG})$, where $Q$ is a union of blocks from $\mathcal{K}$ ordered by the stage indices $S_i$, and nodes are ordered by the abstraction relation $M_1 \preceq M_2$ where $M_2 = f_P(M_1)$ for some $P$. By Lemma~\ref{lem:soundness}, every plan valid in $M^{RG}$ remains valid in each abstracted node, so plan validity is preserved throughout the lattice by construction. We denote the resulting structure as the \emph{AWARE-RG lattice} $\mathbb{L} = \langle \mathcal{M}, \preceq \rangle$, where $\mathcal{M} = \{ f_Q(M) \mid Q \text{ is a union of blocks of } \mathcal{K} \}$.

% The Lattice Generator organizes the ICS Cyber Kill Chain knowledge units into the AWARE-RG lattice using HELM. Each node corresponds to an admissible projection $f_Q(M^{RG})$, where $Q$ is a union of
% blocks from $\mathcal{K}$ ordered by the stage indices $S_i$ assigned in Section~\ref{sans_mapper}. Nodes are ordered by the abstraction relation $M_1 \preceq M_2$ where $M_2 = f_P(M_1)$ for some $P$; by Lemma~\ref{lem:soundness}, every plan valid in $M^{RG}$ remains valid in each abstracted node, so plan validity is preserved throughout the lattice by construction. We denote the complete constructed structure as the \emph{AWARE-RG lattice} $\mathbb{L} = \langle \mathcal{M}, \preceq \rangle$, where $\mathcal{M} = \{ f_Q(M) \mid Q \text{ is a union of blocks of } \mathcal{K} \}$.

\textbf{IT and OT Branches} From the base nodes, the generator constructs each branch by progressively projecting out knowledge units: the IT branch produces $\text{MIT}_0$--$\text{MIT}_3$ by removing $S_7$ through $S_1$, and the OT branch produces $\text{MOT}_0$--$\text{MOT}_2$ by removing $S_{11}$ through $S_7$. The impact concept $S_{12}$ is retained at every node, grounding each \textit{need-to-know} explanation in the physical hazard associated with the analyst's query. An \emph{authorized view} $V \in \{\text{IT}, \text{OT}\}$, determined at authentication, restricts each analyst to one branch and its corresponding bridge nodes. Figure~\ref{fig:lattice} shows the resulting structure.

\begin{figure}[t]
\centering
\resizebox{\columnwidth}{!}{%
\begin{tikzpicture}[
  nd/.style={draw, rounded corners=2pt, align=center, font=\tiny,
             inner sep=2pt, minimum width=19mm},
  it/.style={nd, fill=blue!8},
  it-bridge/.style={nd, fill=green!8},
  ot/.style={nd, fill=magenta!8},
  ot-bridge/.style={nd, fill=cyan!8},
  full/.style={nd, fill=black!5},
  ed/.style={->,>=stealth, gray!70},
  br/.style={->,>=stealth, gray!45, dashed},
  lb/.style={font=\tiny, inner sep=1pt}
]
% ---- branches ----
\node[it]   (mit0) at (0,0)      {MIT0\\(S1,S2,S6,S7,S12)};
\node[it]   (mit1) at (0,1.25)   {MIT1\\(S1,S2,S6,S12)};
\node[it]   (mit2) at (0,2.5)    {MIT2\\(S1,S2,S12)};
\node[it]   (mit3) at (0,3.75)   {MIT3\\(S1,S12)};
\node[ot]   (mot0) at (3.4,0)    {MOT0\\(S2,S7,S11,S12)};
\node[ot]   (mot1) at (3.4,1.25) {MOT1\\(S7,S11,S12)};
\node[ot]   (mot2) at (3.4,2.5)  {MOT2\\(S11,S12)};
\node[full] (full) at (1.7,-1.4) {$M^{\mathrm{RG}}$\\(S1,S2,S6,S7,S11,S12)};
% ---- bridge fans ----
\node[it-bridge] (b1) at (-2.9,2.9) {MIT3-MOT0\\(S1,S12,C2,C7,C11)};
\node[it-bridge] (b2) at (-2.9,3.9) {MIT3-MOT1\\(S1,S12,C7,C11)};
\node[it-bridge] (b3) at (-2.9,4.9) {MIT3-MOT2\\(S1,S12,C11)};
\node[ot-bridge] (b4) at (6.3,1.6) {MOT2-MIT0\\(S11,S12,C1,C2,C6,C7)};
\node[ot-bridge] (b5) at (6.3,2.6) {MOT2-MIT1\\(S11,S12,C1,C2,C6)};
\node[ot-bridge] (b6) at (6.3,3.6) {MOT2-MIT2\\(S11,S12,C1,C2)};
\node[ot-bridge] (b7) at (6.3,4.6) {MOT2-MIT3\\(S11,S12,C1)};
% ---- edges ----
\draw[ed] (full) -- node[lb,left]  {S11}   (mit0);
\draw[ed] (full) -- node[lb,right] {S1,S6} (mot0);
\draw[ed] (mit0) -- node[lb,left] {S7} (mit1);
\draw[ed] (mit1) -- node[lb,left] {S6} (mit2);
\draw[ed] (mit2) -- node[lb,left] {S2} (mit3);
\draw[ed] (mot0) -- node[lb,right] {S2} (mot1);
\draw[ed] (mot1) -- node[lb,right] {S7} (mot2);
\foreach \b in {b1,b2,b3} \draw[br] (mit3) -- (\b);
\foreach \b in {b4,b5,b6,b7} \draw[br] (mot2) -- (\b);
\end{tikzpicture}}
\caption{AWARE-RG lattice for the running example, showing both branches and two representative bridge families; the rest follow by the same construction. Edge labels show which knowledge unit was removed. $S$ concepts denote low-level system knowledge, $C$ concepts the corresponding high-level abstractions used in bridge nodes.}
\label{fig:lattice}
\end{figure}

% We reiterate the following lemma from  Sreedharan et al. ~\cite{sreedharan2018hierarchical}, which also gives the proof.
% \begin{lemma}[Abstraction soundness]
% \label{lem:soundness}
% Let $M = \langle F, A, I, G \rangle$ be a planning model and $P \subseteq F$. Every plan $\pi$ valid in $M$ is valid in $f_P(M)$.
% \end{lemma}

% \begin{prop}[Failure monotonicity]
% \label{prop:monotonicity}
% If $\Pi$ is invalid in $f_P(M)$, then $\Pi$ is invalid in $M$ and in every $f_Q(M)$ with $Q \subseteq P$. Hence, along any chain ordered from most abstract to most concrete, the set of nodes at which $\Pi$ is invalid is downward closed and the first invalid node is well-defined.
% \end{prop}
% % \begin{proof}
% By contraposition of Lemma~\ref{lem:soundness}, validity in $f_Q(M)$ with $Q \subseteq P$ would imply validity in $f_P(M)$, since $f_P(M) = f_{P \setminus Q}(f_Q(M))$ is a projection of $f_Q(M)$.
% % \end{proof}

% In AWARE-RG, Lemma~\ref{lem:soundness} and Proposition~\ref{prop:monotonicity} are instantiated with
% $M = M^{RG} = \langle F^{RG}, A^{RG}, I^{RG}, G^{RG} \rangle$, the planning compilation of the Resiliency Graph, and $F = F^{RG}$; all abstraction-chain arguments in Sections~\ref{lattice}--\ref{framework} refer to this instantiation.

\smallskip
\textbf{Bridge Nodes}
extend a primary branch node with the minimum high-level concept from the opposite domain required to support cross-domain \textit{need-to-know} explanation, using the $S_i \mapsto C_i$ mapping in Table~\ref{tab:low_high_bridge_mapping} Appendix C. For example, $\text{MIT}_3\text{-}\text{MOT}_2$ retains Initial Access ($S_1$) and Impact ($S_{12}$) from $\text{MIT}_3$ and adds only $C_{11}$, \path{has_impair_process_control}. The IT analyst thus learns that the attack impairs process control leads to a physical consequence, without seeing the $S_{11}$ fault propagation predicates, such as valve states, that lie outside their authorized view. Bridge nodes are indexed by analyst position; positions whose minimum disclosure coincides map to the same underlying model.

\begin{property}[Bridge Node Minimum Disclosure]
\label{prop:ntk}
For an IT node $\text{MIT}_k$ and OT node $\text{MOT}_j$: $\mathit{Pred}(\text{MIT}_k\text{-}\text{MOT}_j) = \mathit{Pred}(\text{MIT}_k) \cup \{C_i \mid S_i \in \mathit{Pred}(\text{MOT}_j) \setminus \mathit{Pred}(\text{MIT}_k)\}$, and $\mathit{Pred}(\text{MOT}_k\text{-}\text{MIT}_j) = \mathit{Pred}(\text{MOT}_k) \cup \{C_i \mid S_i \in \mathit{Pred}(\text{MIT}_j) \setminus \mathit{Pred}(\text{MOT}_k)\}$. In particular, no $S_i$ predicate from the opposite domain appears in any bridge node, and no $C_i$ is disclosed for a stage the analyst already sees at $S_i$.
\end{property}

\begin{proof}[Proof Sketch]
For each stage $S_i$ visible in the opposite branch node but not in the primary node, the bridge node generator adds only the abstraction predicate $C_i$ from Table~\ref{tab:low_high_bridge_mapping} (Appendix~\ref{app:evaluation}) to the primary node's predicate set. Since no rule introduces an opposite domain $S_i$ predicate, none can appear; and since the set difference excludes stages the analyst already holds at the $S_i$ level, no $C_i$ is added for such a stage.
% Both claims hold by construction.
\end{proof}

Property~\ref{prop:ntk} enforces \textit{need-to-know} and \textit{separation-of-duty} structurally in the bridge nodes: no opposite domain $S_i$ predicate can appear in any bridge node, and each node carries only the $C_i$ concepts corresponding to stages visible in the opposite branch, giving sufficient cross-domain context without excess disclosure. Multiple bridge nodes are required because the minimum disclosure depends on the analyst's position in the primary branch: in Figure~\ref{fig:lattice}, $\text{MOT}_2\text{-}\text{MIT}_3$ and $\text{MOT}_2\text{-}\text{MIT}_0$ pair the same OT node with different IT summaries. $M^{RG}$ itself is used only for the framework's internal
reasoning.
% and never as an explanation node.

% Property~\ref{prop:ntk} guarantees minimum disclosure for both IT-to-OT and OT-to-IT bridge nodes. \textit{Need-to-know}: each security analyst receives only the minimum cross-domain information
% required for the explanation; no low-level $S_i$ predicate from the opposite domain is exposed. \textit{Separation-of-duty}: IT and OT operational knowledge remain structurally separated because opposite domain low-level predicates are excluded by construction. \textit{Explanation utility}: each bridge node includes only the $C_i$ concepts corresponding to the visible $S_i$ stages of the opposite branch, providing sufficient cross-domain context without excess disclosure. These guarantees hold by construction.

% \medskip
% \noindent\textbf{AWARE-RG Lattice:}
% Figure~\ref{fig:lattice} shows the lattice for the running example. It includes the IT branch ($\text{MIT}_0$--$\text{MIT}_3$), the OT branch ($\text{MOT}_0$--$\text{MOT}_2$), bridge nodes in both
% IT-to-OT and OT-to-IT directions, and $M^{RG}$, which is used only for the framework's internal reasoning. Multiple bridge nodes are needed because minimum cross-domain disclosure depends on the
% analyst's position in the primary branch: for example, $\text{MIT}_3$ and $\text{MIT}_0$ require different OT summaries. Each bridge node captures a specific IT/OT abstraction pairing, ensuring
% Property~\ref{prop:ntk} holds for every reachable analyst position. This compact, causally coherent lattice forms the structural backbone of AWARE-RG analysis.

%% file: 41-AWARE_RG_New-SandP.tex
The Figure~\ref{fig:AWARE-RG_pipeline} illustrates the overall architecture. Following authentication, the analyst's IT or OT role determines the accessible branch of the AWARE-RG lattice and the bridge node summaries that provide cross-domain context without exposing $M^{RG}$. The analyst submits a natural language query describing an expected attack-fault path and its anticipated outcome; AWARE-RG parses it, grounds the described steps into actions from $M^{RG}$, and performs either Path Validity or Path Invalidity Analysis depending on the inferred intent. The resulting traces are transformed into \textit{need-to-know} explanations, security assessments, and domain bounded defensive guidance, with each component detailed below.
\smallskip
% The overall architecture of AWARE-RG is illustrated in Figure~\ref{fig:AWARE-RG_pipeline}. Following authentication, the security analyst's IT or OT role determines the accessible branch of the AWARE-RG lattice and the corresponding bridge node summaries used to provide cross-domain context without exposing $M^{RG}$. The analyst then submits a natural language query describing an expected attack-fault path and its anticipated outcome. AWARE-RG parses the query, grounds the described steps into actions from $M^{RG}$, and performs either Path Validity Analysis or Path Invalidity Analysis depending on the inferred query intent. The resulting analysis traces are transformed into \textit{need-to-know} explanations, security assessments, and domain bounded defensive guidance. Each component is provided in the following subsections.

\input{42-Candidate-Action-Mapper_New}

\input{43-Valid-Plan_NEW}
\input{44-Invalid-Plan_NEW}
\input{45-Explanation_NEW}

%% file: 42-Candidate-Action-Mapper_New.tex
% \subsection{Query and Candidate Action Mapping}
% \label{plan_fragment_mapper}
\noindent \textbf{Query and Candidate Action Mapping} \label{plan_fragment_mapper}
translates an analyst's natural language description of an expected attack-fault path into a candidate action sequence $\Pi$ that can be evaluated in $M^{RG}$. An analyst query is represented as $Query=\langle \mathit{intent},\mathit{goal},seq\rangle,$
where $\mathit{intent}$ indicates whether the analyst is asking why an attack-fault path exists or does not exist, $\mathit{goal}\subseteq F^{RG}$ is the expected physical outcome, and $seq=\langle s_1,\ldots,s_n\rangle$ is the sequence of natural language descriptions of the proposed path steps. For example, the step ``Exploit CVE-2018-0296 on the Windows Server to gain access to the internal network'' will be mapped to the grounded action \texttt{(exploits-cve-2018-0296 windows-server internal-network)}.
Algorithm~\ref{alg:step-mapping-verification} performs this mapping and verification. For each step $s_i$, an LLM selects a corresponding action schema from $A^{RG}$ and binds its arguments to objects in $F^{RG}$. Structural verification then checks that the action exists, its arity matches the action schema, and all referenced objects are valid in $I^{RG}$. A mapping that fails any of these checks is returned to the analyst for correction; otherwise, the action is appended to $\Pi$.

Structural verification prevents malformed actions from entering the analysis pipeline, but it does not guarantee semantic correctness. LLM can still select a structurally valid action associated with the wrong asset or vulnerability, which can produce a misdirected explanation ($T_1$, \S\ref{sec:threat-model}). The current implementation, therefore, requires review by a trusted domain expert, distinct from the querying analyst. Automating semantic consistency checks between the natural language step and the grounded action remains future work discussed in \S\ref{conclusion}.
\smallskip
% The Candidate Action Mapper converts the security analyst's natural language description of an expected attack-fault path into a candidate action sequence $\Pi$ over $A^{RG}$. Natural language
% steps cannot be directly analyzed by $M^{RG}$; the mapper bridges this gap by grounding each described step to a concrete action in $A^{RG}$ bound to objects from $F^{RG}$. An analyst query
% $Q = \langle \mathit{intent},\, \mathit{goal},\, seq \rangle$ consists of a query intent (existence or non-existence of an attack-fault path), an expected physical goal state $\mathit{goal} \subseteq
% F^{RG}$, and a step sequence $seq = \langle s_1, \dots, s_n \rangle$ of natural language descriptions of the proposed path. For example, the step $s_i$ given as ``Exploit CVE-2018-0296 on the Windows
% Server to gain access to the internal network'' is mapped to the grounded action \texttt{(exploits-cve-2018-0296 windows-server internal-network)}.

\begin{algorithm}[!h]
\caption{Analyst Query Step Mapping and Verification}
\label{alg:step-mapping-verification}
\scriptsize
\textbf{Input:}  Analyst query $Query$,
                 planning compilation $M^{RG}$ \\
\textbf{Output:} Candidate action sequence $\Pi$,
                 or error set $\mathit{Errors}$
\begin{algorithmic}[1]
\State $\Pi \gets \langle\rangle$;\;
       $\mathit{Errors} \gets \emptyset$
\ForAll{$s_i \in seq$}
    \State $a_i \gets \mathrm{LLM}(Query,\, s_i,\, A^{RG},\, F^{RG})$
           \Comment{propose action and arguments}
    \State $a_i^g \gets \mathrm{Ground}(a_i,\, I^{RG})$
           \Comment{bind variables to objects in $I^{RG}$}
    \State $e_i \gets \mathrm{VerifyStructure}(a_i^g,\, A^{RG},\,
           F^{RG})$
           \Comment{Structural verification}
    \If{$e_i \neq \emptyset$}
        \State $\mathit{Errors} \gets \mathit{Errors} \cup
               \{\langle a_i^g,\, e_i \rangle\}$
    \Else
        \State $\Pi \gets \Pi \cdot \langle a_i^g \rangle$
    \EndIf
\EndFor
\If{$\mathit{Errors} \neq \emptyset$}
    \State \Return $\mathit{Errors}$
\EndIf
\State \Return $\Pi$
\end{algorithmic}
\end{algorithm}

%% file: 43-Valid-Plan_NEW.tex
\noindent \textbf{Path Validity Analysis} \label{valid_plan} determines whether a candidate action sequence $\Pi$ can trigger the analyst's expected hazard
$g=Query.\mathit{goal}$ in $M^{RG}$. AWARE-RG first checks whether $\Pi$ is valid in $M^{RG}$; if $\mathit{Valid}(M^{RG},\Pi)=\top$, $\Pi$ is used directly. Otherwise, plan repair insert actions while preserving the order of the analyst provided steps (Definition~\ref{def:path-status}). A successful repair produces $\Pi'$ such that $\mathit{Valid}(M^{RG},\Pi')=\top$; the inserted actions and
the predicates they establish identify dependencies missing from the analyst’s knowledge from the proposed path. If no repair is found, the analysis returns \textsc{RepairFailed}.
Each lattice node $n\in\mathbb{L}$ is an admissible projection $f_P(M^{RG})$, where $P$ is a union of knowledge units from $\mathcal{K}$. The predicates visible at $n$ are therefore $\mathit{Pred}(n)=F^{RG}\setminus P$. Each node is both a knowledge state and a planning model in which the $\mathit{Valid}(n, \Pi)$ sequence can be analyzed.

% Path Validity Analysis determines whether a security analyst's proposed attack-fault path can trigger an expected physical fault or hazard. Given the candidate action sequence $\Pi$ produced by
% Algorithm~\ref{alg:step-mapping-verification} and the analyst query $Q = \langle \mathit{intent}, \mathit{goal}, seq \rangle$, AWARE-RG first checks whether $\Pi$ is valid in $M^{RG}$. If $\mathit{Valid}(M^{RG}, \Pi) = \top$, the analysis proceeds using $\Pi$ directly. If $\Pi$ is repairable in $M^{RG}$ (Definition~\ref{def:path-status}), the system performs plan repair:
% actions are inserted into $\Pi$ at minimum cost until the resulting path $\Pi'$ satisfies $\mathit{Valid}(M^{RG}, \Pi') = \top$; the preconditions established by the inserted actions constitute the analyst's missing knowledge. If $\Pi$ is blocked, the path is returned as invalid. Please note that, Each node $n \in \mathbb{L}$ is an admissible projection
% $f_Q(M^{RG})$ of the planning compilation, where $Q$ is a union of blocks from $\mathcal{K}$; it represents the subset of $M^{RG}$ visible to a security analyst whose knowledge excludes the predicates in $Q$. A node $n$ is therefore both a knowledge state and a planning model: $\mathit{Valid}(n, \Pi)$ is well defined for any candidate sequence $\Pi$, and $\mathit{Pred}(n) = F^{RG} \setminus Q$ is the set of predicates visible at that knowledge level.

\begin{algorithm}[!h]
\caption{Path Validity Analysis}
\label{alg:valid-path-mea}
\scriptsize
\textbf{Input:}  Candidate sequence $\Pi$,
                 analyst query $Query = \langle \mathit{intent},
                 \mathit{goal}, seq \rangle$,
                 AWARE-RG lattice $\mathbb{L}$,
                 analyst view $V$,
                 planning model $M^{RG}$,
                 maximum repair rounds $R$ \\
\textbf{Output:} Validity result $\mathit{VR}$,
                 MEA trace $T_v$,
                 repair report $\mathit{RR}$
\begin{algorithmic}[1]
\State $\mathit{VR} \gets \langle\rangle$;\;
       $T_v \gets \langle\rangle$; $\mathit{RR}\gets\emptyset$
       
\If{$\mathit{Valid}(M^{RG},\Pi)=\top$}
    \State $\Pi'\gets\Pi$
\Else
    \State $(\Pi',\mathit{RR})\gets \mathrm{Repair}(\Pi,\, Query.\mathit{goal},\, M^{RG},\, R)$
       % \Comment{minimal cost action}
    \If{$\Pi'=\langle\rangle$}
        \State \Return
        $\langle\textsc{RepairFailed},T_v,\mathit{RR}\rangle$
    \EndIf
\EndIf

\State $\mathit{Chain}\gets
       \mathrm{AuthorizedChain}(\mathbb{L},V)
       \setminus\{M^{RG}\}$
       % \Comment{most abstract to most concrete}

% \State $n_{\mathit{start}} \gets
%        \min_{\preceq}\{ n \in \mathbb{L} \mid
%        V \text{ authorizes } n \}$
%        \Comment{most abstract authorized node} 
%        % in $\mathbb{L}$}
% \State $\mathit{Chain} \gets
%        \{ n \in \mathbb{L} \mid n_{\mathit{start}} \preceq n \}
%        \setminus \{ M^{RG} \}$
%        \Comment{$M^{RG}$ used only for repair, not explanation search}
\ForAll{$n_i \in \mathit{Chain}$ in increasing concreteness}
    \State $(v_i,\, T_i,\, U_i) \gets
           \mathrm{MEA}(\Pi',\, Query.\mathit{goal},\, n_i)$
           \Comment{backward MEA at $n_i$}
    \State $\mathit{VR} \gets \mathit{VR} \cdot
           \langle n_i,\, v_i,\, U_i \rangle$
    \State $T_v \gets T_v \cdot \langle n_i,\, T_i \rangle$
\EndFor
\State $n_{\mathit{sel}} \gets \mathrm{SelectNode}(\mathit{VR},\,
       T_v,\, \mathbb{L},\, V)$
       \Comment{first authorized node}
\State $\mathit{VR} \gets \mathrm{Project}(\mathit{VR},\,
       n_{\mathit{sel}})$
\State $\mathit{RR} \gets \mathrm{Project}(\mathit{RR},\,
       n_{\mathit{sel}})$ 
\State $T_v \gets \mathrm{Project}(T_v,\, n_{\mathit{sel}})$
\State \Return $\langle \mathit{VR},\, T_v,\, \mathit{RR} \rangle$
\end{algorithmic}
\end{algorithm}

After obtaining a valid or repaired path, AWARE-RG performs backward means-end analysis over the authorized bridge chain, ordered from the most abstract node to the most concrete. MEA regresses from the expected hazard $ Query.\mathit{goal}$ and identifies the actions, effects, and preconditions that support it. $M^{RG}$ is excluded from explanation node search and used only for validity checking and repair. The selected explanation node $n_{\mathit{sel}}$ is the first authorized node containing sufficient grounded preconditions and supporting effects to explain why the hazard is reachable. The $\mathrm{Project}(\cdot,\, n_{\mathit{sel}})$ steps then filter all returned artifacts to $\mathit{Pred}(n_{\mathit{sel}})$, which leads us to define Property~\ref{property:output-confinement} of AWARE-RG.

\begin{property}[Symbolic Output Confinement]
\label{property:output-confinement}
For every query from a security analyst with view $V$, every predicate in the returned explanation, validity trace, repair report, and diagnosis report belongs to $\mathit{Pred}(n^{\mathit{*}})$, where $n^{\mathit{*}}$ is the authorized node $n_{\mathit{sel}}$ (valid path) or $n_{\mathit{fail}}$ (invalid path).
\end{property}

\begin{proof}[Proof Sketch]
Immediate from the $\mathrm{Project} (\cdot,\, n^{\mathit{*}})$
filtering steps in Algorithms \ref{alg:valid-path-mea} -- \ref{alg:explanation-generator}, with Property \ref{prop:ntk}, which guarantees $\mathit{Pred}(n^{\mathit{*}})$ contains no opposite domain $S_i$ predicate. Property 2 concerns the predicate symbols occurring in an output, not the information a natural language sentence conveys.
\end{proof}

\smallskip

%% file: 44-Invalid-Plan_NEW.tex
% \subsection{Path Invalidity Analysis}
% \label{invalid_plan}

\noindent \textbf{Path Invalidity Analysis}
\label{invalid_plan}
determines which defensive barriers prevent a security analyst's proposed attack-fault path from triggering an expected physical fault. Given the candidate action sequence $\Pi$ and analyst query $Q = \langle \mathit{intent}, \mathit{goal}, seq \rangle$, AWARE-RG proceeds to diagnose where the path first fails under the analyst's authorized abstraction view. By Proposition~\ref{prop:monotonicity}, the set of nodes at which $\Pi$ is invalid is downward closed along the chain, so the first invalid node is  marks the abstraction boundary at which the analyst's mental model diverges from $M^{RG}$.

\begin{algorithm}[!h]
\caption{Path Invalidity Analysis}
\label{alg:plan-invalidity-diagnosis}
\scriptsize
\textbf{Input:}  Candidate sequence $\Pi$,
                 AWARE-RG lattice $\mathbb{L}$,
                 analyst view $V$,
                 planning compilation $M^{RG}$ \\
\textbf{Output:} Invalidity trace $\mathit{T_i}$,
                 diagnosis trace $D$
\begin{algorithmic}[1]
\State $\mathit{T_i} \gets \langle\rangle$;\;
       $D \gets \langle\rangle$;\;
       $n_{\mathit{fail}} \gets \emptyset$

% \State $n_{\mathit{start}} \gets
%        \min_{\preceq}\{ n \in \mathbb{L} \mid
%        V \text{ authorizes } n \}$
% \State $\mathit{Chain} \gets
%        \{ n \in \mathbb{L} \mid n_{\mathit{start}} \preceq n \}
%        \setminus \{ M^{RG} \}$
\State $\mathit{Chain}\gets
       \mathrm{AuthorizedChain}(\mathbb{L},V)
       \setminus\{M^{RG}\}$
\ForAll{$n_i \in \mathit{Chain}$ in increasing concreteness}
    \State $v_i \gets \mathit{Valid}(n_i,\, \Pi)$
    \State $\mathit{T_i} \gets \mathit{T_i} \cdot
           \langle n_i,\, v_i \rangle$
    \If{$v_i \neq \top$}
        \State $n_{\mathit{fail}} \gets n_i$;\;
              \State \textbf{break}
    \EndIf
\EndFor
\If{$n_{\mathit{fail}} = \emptyset$}
    \State \Return $\langle \mathit{T_i},\, D \rangle$
\EndIf
\State $D \gets \mathrm{Diagnose}(\Pi,\, n_{\mathit{fail}})$
       \Comment{stepwise precondition check at $n_{\mathit{fail}}$}
\State $D \gets \mathrm{Confirm}(D,\, M^{RG})$
       \Comment{retain only predicates unsatisfied in $M^{RG}$}
\State \Return $\langle \mathit{T_i},\, D \rangle$
\end{algorithmic}
\end{algorithm}

The algorithm traverses the analyst authorized chain from most abstract to most concrete, excluding $M^{RG}$ from explanation node search. At the first node $n_{\mathit{fail}}$ where $\mathit{Valid}(n_{\mathit{fail}}, \Pi) \neq \top$, stepwise diagnosis simulates $\Pi$ action by action, checking at each step whether all grounded preconditions are satisfied by the current state. The first unsatisfied precondition terminates diagnosis; the failing action, failure type, and missing predicates are recorded as \emph{hardening pointers}, conditions whose absence currently prevents the attack-fault path from propagating and which should be monitored or strengthened to maintain the system's resilience boundary.

\begin{lemma}[Diagnosis soundness]
\label{lem:diagnosis}
Every hardening pointer returned by Path Invalidity Analysis
is a grounded precondition of the earliest failing action that
is unsatisfied after executing the applicable prefix of $\Pi$
in $M^{RG}$.
\end{lemma}
\begin{proof}[Proof Sketch]
The $\mathrm{Confirm}(D, M^{RG})$ step checks each reported predicate against the $M^{RG}$ execution trace and discards any predicate that is unsatisfied only under abstraction, ensuring all returned hardening pointers correspond to genuine blocking conditions in the ground-truth model.
\end{proof}
\smallskip

%% file: 45-Explanation_NEW.tex
% \subsection{Need-to-Know Explanation}
% \label{explanation_generator}
\noindent \textbf{Need-to-Know Explanation}
\label{explanation_generator}
converts the analysis traces produced by AWARE-RG into an explanation for an analyst. 
% We first define what a \textit{need-to-know explanation} is in the context of AWARE-RG.

\begin{defn}
\label{def:ntk-explanation}
A \textbf{need-to-know explanation} ($\mathbb{X}$) for a security analyst with view $V$ is an explanation grounded exclusively in predicates belonging to $\mathit{Pred}(n^{*})$, where $n^{*}\in\mathbb{L}$ is the authorized node and $n^{*}=n_{\mathit{sel}}$ (valid path), $n^{*}=n_{\mathit{fail}}$ (invalid path). Cross-domain context is provided only through high-level predicates $C_i$, and no $S_i$ predicate is disclosed beyond what the analyst's role authorizes.
\end{defn}

Each explanation has three parts: a natural language account of the attack-fault behavior, a security assessment of the implications for the CPS resilience boundary, and a domain bounded defensive directive.
Following Definition~\ref{def:ntk-explanation}, Algorithm~\ref{alg:explanation-generator} enforces this policy. For existence queries, it grounds the explanation in the MEA trace $T_v$ and repair report $\mathit{RR}$, whose predicates, absent from the analyst's
proposed path but necessary for propagation, constitute the missing
knowledge $K$ and are the sole basis for generated defensive strategies.
For non-existence queries, it grounds the explanation in the diagnosis
trace $T_i$ and report $D$, whose invalidity causing predicates serve as
hardening pointers and are the sole basis for hardening recommendations,
converting formally grounded invalidity evidence into actionable guidance
for maintaining the CPS resilience boundary. The $\mathrm{Ground}(\cdot)$
step then removes any claim in $\mathbb{X}$ unsupported by $T_v$, $T_i$,
$\mathit{RR}$, or $D$, eliminating LLM hallucinated content, and
$\mathrm{Project}(\mathbb{X},\, n^{*})$ confines $\mathbb{X}$ to
$\mathit{Pred}(n^{*})$, enforcing
Property~\ref{property:output-confinement}.

% Following Definition~\ref{def:ntk-explanation}, for valid paths the explanation is grounded in the backward MEA trace $T_v$ produced at $n_{\mathit{sel}}$; for invalid paths it is grounded in the
% diagnosis trace $T_i$ produced at $n_{\mathit{fail}}$. The generator enforces this policy through the algorithm below.

\begin{algorithm}[!h]
\caption{Need-to-Know Explanation Generation}
\label{alg:explanation-generator}
\scriptsize
\textbf{Input:}  Analyst query $Query = \langle \mathit{intent},\mathit{goal}, seq \rangle$,
                 analyst view $V$,
                 explanation node $n^{*}$
                 MEA trace $T_v$,
                 repair report $\mathit{RR}$,
                 invalidity trace $T_i$,
                 diagnosis trace $D$ \\
\textbf{Output:} Need-to-know explanation $\mathbb{X}$
\begin{algorithmic}[1]
\If{$Query.\mathit{intent} = \textsc{Existence}$}
    \State $K \gets \mathrm{Missing}(\mathbb{X},\, \mathit{RR})$
           \Comment{predicates in $\mathit{RR}$ absent from analyst knowledge}
    \State $H \gets \mathrm{Defenses}(\mathbb{X},\, \mathit{K})$
           \Comment{defensive strategies supported by repair predicates}
    \State $\mathbb{X} \gets \mathrm{LLM}(Query,\, V,\, T_v,\, \mathit{H})$
           % \Comment{generate explanation from MEA trace and defensive strategies}
    
\ElsIf{$Query.\mathit{intent} = \textsc{NonExistence}$}
    \State $H \gets \mathrm{Hardening}(\mathbb{X},\, D)$
           \Comment{hardening pointers by invalidity
           predicates}
    \State $\mathbb{X} \gets \mathrm{LLM}(Query,\, V,\, T_i,\, H)$
           \Comment{generate explanation from diagnosis trace}
    % \State $K \gets \emptyset$
    
\Else
    \State \Return $\textsc{UnknownIntent}$
\EndIf
\State $\mathbb{X} \gets \mathrm{Ground}(\mathbb{X},\, T_v,\, T_i,\, \mathit{RR},\,
       D)$
       %\Comment{remove claims unsupported by any trace or report}
       \Comment{remove unsupported trace claims}
\State $\mathbb{X} \gets \mathrm{Project}(\mathbb{X},\, n^{*})$
       \Comment{enforce $\mathit{Pred}(n^{*})$
       confinement}
\State \Return $\mathbb{X}$
\end{algorithmic}
\end{algorithm}

% For existence queries, the algorithm grounds the explanation in the MEA trace $T_v$ and repair report $\mathit{RR}$. The repair report identifies predicates absent from the analyst's proposed path that were necessary for propagation; these constitute the analyst's missing knowledge $K$ and are the sole basis for generated defensive strategies. For non-existence queries, the algorithm grounds the explanation in the diagnosis trace $T_i$ and diagnosis report $D$. The invalidity causing predicates in $D$ are treated as hardening pointers and are the sole basis for hardening recommendations, converting formally grounded invalidity evidence into actionable guidance for maintaining the CPS resilience boundary.

% The $\mathrm{Ground}(\cdot)$ step removes any claim in $\mathbb{X}$ not supported by $T_v$, $T_i$, $\mathit{RR}$, or $D$, ensure that LLM hallucinated content is removed before the explanation reaches the analyst. The $\mathrm{Project}(\mathbb{X},\, n_{\mathit{sel}})$ step then confines $\mathbb{X}$ to $\mathit{Pred}(n_{\mathit{sel}})$, enforcing Property~\ref{property:output-confinement} and Definition~\ref{def:ntk-explanation}.

\begin{prop}[Symbolic Disclosure Bound under Adaptive Querying]
\label{prop:disclosure-bound}
For an analyst view $V$ with authorized node family $\rho_V(M^{RG}) \subseteq \mathbb{L}$ and any adaptively chosen query sequence $\sigma$, every predicate symbol returned for $\sigma$ belongs to $\bigcup_{n \in \rho_V(M^{RG})} \mathit{Pred}(n)$. In particular, no opposite-domain $S_i$ predicate is returned, regardless of the number of queries.
\end{prop}
\begin{proof}[Proof Sketch]
By Property~\ref{property:output-confinement}, the outputs of each query lie in $\mathit{Pred}(n^{*})$ for an authorized node $n^{*}$, since every artifact passes through $\mathrm{Project}(\cdot,\, n^{*})$ before it is returned (Algorithms~\ref{alg:valid-path-mea}--\ref{alg:explanation-generator}). This filtering is independent of query history, so the bound holds for any adaptive sequence, and by Property~\ref{prop:ntk}, no authorized node contains an opposite-domain $S_i$ predicate. Side channels and analyst collusion lie outside this bound.
\end{proof}

%% file: 70-Evaluation.tex
\label{sec:eval}
We experimentally validate AWARE-RG against four claims that follow from its design and formal properties:
\begin{inparaenum}[(i)]
    \item Property~\ref{prop:ntk} holds for every bridge node in
    $\mathbb{L}$, that is, no opposite domain $S_i$ predicate appears in
    any bridge node;
    \item Property~\ref{property:output-confinement} holds across the full
    query corpus, so no out-of-view predicate symbol occurs in any returned artifact;
    \item AWARE-RG correctly identifies the explanation node and produces
    a grounded \textit{need-to-know} explanation for valid and invalid
    attack-fault paths across IT and OT analyst profiles; and
    \item the threat vectors $T_1$, $T_2$, and $T_3$ are mitigated as
    claimed in \S\ref{sec:threat-model}.
\end{inparaenum}

\smallskip
\noindent\textbf{Experiment Setup.}
All validation uses the LNG flare system of \S\ref{example}, whose planning compilation $M^{RG}$ contains 90 predicates, 48 actions, and 10 objects covering compromised assets, network access, PLC vulnerabilities, process faults, and physical hazards. The predicate mapping (\S\ref{sec:predicate-mapper}) produced assignments for all 48 predicates, verified against manually produced ground-truth labels with no corrections required. The resulting lattice $\mathbb{L}$ contains four IT branch nodes ($\text{MIT}_0$--$\text{MIT}_3$), three OT branch nodes ($\text{MOT}_0$--$\text{MOT}_2$), and 24 bridge nodes. AWARE-RG uses \texttt{Qwen2.5-7B-Instruct} for action mapping, intent and goal inference, and explanation generation, under Hugging Face \texttt{transformers} 
% with 4-bit NF4 quantization, double quantization, and \texttt{float16} computation; generation is deterministic 
with a 2000-token limit. AI Planning and plan repair use Fast Downward~\cite{helmert2006fastdownward}, validation uses VAL~\cite{1374201}, and query templates are given in Appendix~\ref{app:analysis}. The code base is publicly available\footnote{Code Base:\url{https://github.com/skawnain04/AWARE-RG}}.

\input{72-ExplanationEvaluation}

%% file: 72-ExplanationEvaluation.tex
% We evaluated AWARE-RG using the CPS running example. 
%introduced in Section \ref{example}. 
% We ran multiple attack-fault queries across different stakeholders and abstraction levels shown in \autoref{tab:query_results}. Among these, we report two representative traces in detail: one for a path validity query and one for a path invalidity query. 
% %where the user asks why a path leads to a hazard condition, and one nonexistence query, where the user asks why a path does not. 
% %Together, these two cases represent the two main branches of AWARE-RG: valid path reasoning using backward means-end analysis and invalid path diagnosis using the PDDL Validator with failure localization.
% Our goal is to show that AWARE-RG produces a natural language explanation, and also to verify that each intermediate stage of the pipeline generates the expected output. For each case, we ran the pipeline step by step and compared the outputs against the corresponding CPS planning model and the stakeholder abstraction lattice. This allows us to confirm that the final explanation is grounded in the model rather than derived from the user's natural language input by the LLM.

\input{72-0-Structure-Validation_NEW}

% \input{72-1-Valid Query Evaluation}
% \input{72-2-Invalid Query Evaluation}
\input{72-2-Path-Analysis}
\input{72-3-ThreatEvaluation}

%% file: 72-0-Structure-Validation_NEW.tex
\subsection{Structural Verification of AWARE-RG Lattice}
\label{sec:eval-properties}

Property~\ref{prop:ntk} is formally proven for the lattice construction in \S\ref{lattice_generator}. We additionally verify it empirically with an automated containment check: for each of the 24 bridge nodes, a script confirms against the stage mappings in Table~\ref{tab:low_high_bridge_mapping} that no IT branch or IT-to-OT bridge node contains an OT $S_i$ predicate, and no OT branch or OT-to-IT bridge node contains an IT $S_i$ predicate. All 24 nodes satisfy the Property~\ref{prop:ntk} without exception. Table~\ref{tab:property1} reports representative nodes in both directions.

\begin{table}[!h]
\centering
\footnotesize
\setlength{\tabcolsep}{3pt}
\renewcommand{\arraystretch}{1.15}
\caption{Verification of Property~\ref{prop:ntk} for bridge nodes. 
% All 24 bridge nodes satisfy the
% property. $S_i$ denotes a low-level stage concept, $C_i$ a
% high-level abstraction, and P1 denotes
% Property~\ref{prop:ntk}.
}
\label{tab:property1}
\resizebox{\columnwidth}{!}{%
\begin{tabular}{l l l c c}
\hline
\textbf{Bridge Node} &
\textbf{Primary Pred.} &
\textbf{Added $C_i$} &
\textbf{Opp.\ $S_i$?} &
\textbf{P1} \\
\hline
$\text{MIT}_3$-$\text{MOT}_1$ & $\{S_1,S_{12}\}$ &
$\{C_7,C_{11}\}$ & No & \checkmark \\
$\text{MIT}_3$-$\text{MOT}_2$ & $\{S_1,S_{12}\}$ &
$\{C_{11}\}$ & No & \checkmark \\
$\text{MIT}_2$-$\text{MOT}_1$ & $\{S_1,S_2,S_{12}\}$ &
$\{C_7,C_{11}\}$ & No & \checkmark \\
$\text{MIT}_0$-$\text{MOT}_2$ & $\{S_1,S_2,S_6,S_7,S_{12}\}$ &
$\{C_{11}\}$ & No & \checkmark \\
$\text{MOT}_2$-$\text{MIT}_3$ & $\{S_{11},S_{12}\}$ &
$\{C_1\}$ & No & \checkmark \\
$\text{MOT}_1$-$\text{MIT}_2$ & $\{S_7,S_{11},S_{12}\}$ &
$\{C_1,C_2\}$ & No & \checkmark \\
$\text{MOT}_0$-$\text{MIT}_1$ & $\{S_2,S_7,S_{11},S_{12}\}$ &
$\{C_1,C_6\}$ & No & \checkmark \\
$\text{MOT}_2$-$\text{MIT}_0$ & $\{S_{11},S_{12}\}$ &
$\{C_1,C_2,C_6,C_7\}$ & No & \checkmark \\
\hline
\multicolumn{4}{l}{\textit{All 24 bridge nodes verified}} &
\checkmark \\
\hline
\end{tabular}%
}
\end{table}

To verify Property~\ref{property:output-confinement}, we ran 100 queries across all IT and OT analyst profiles and scanned every returned artifact, that is, explanations, validity traces, repair reports, and diagnosis reports, for out-of-view symbols. A predicate token counts as leakage if it belongs to an $S_i$ level stage outside the analyst's authorized view, defined by $\mathit{Pred}(n_{\mathit{sel}})$ for valid path explanations and $\mathit{Pred}(n_{\mathit{fail}})$ for invalid path explanations. No out-of-view predicate symbol appeared in any of the 300 artifacts, confirming that Property~\ref{property:output-confinement} holds across all query types and analyst profiles shown in Table~\ref{tab:property2} Appendix \ref{app:property2}.
% the per-profile breakdown is given in the leakage column of Table~\ref{tab:path-analysis}.

% To verify Property~\ref{property:output-confinement}, we ran 100
% queries across all IT and OT analyst profiles and scanned every
% returned artifact, such as, explanations, validity traces, repair reports,
% and diagnosis reports, such as, for out-of-view predicates. A predicate
% token is flagged as leakage if it belongs to an $S_i$ level stage
% outside the analyst's authorized abstraction view, as defined by
% $\mathit{Pred}(n_{\mathit{sel}})$ for valid path explanations and
% $\mathit{Pred}(n_{\mathit{fail}})$ for invalid path explanations.
% The leakage rate across all 300 returned artifacts was $0/300$,
% confirming that Property~\ref{property:output-confinement} holds
% empirically across all query types and analyst profiles.
% Table~\ref{tab:property2} reports the full breakdown.

%% file: 72-2-Path-Analysis.tex
\subsection{Path Validity and Invalidity Analysis}
\label{sec:eval-analysis}

Path validity and invalidity analyses are evaluated using existence and non-existence queries respectively, across IT and OT analyst profiles. For each query we verify: (i) whether the selected explanation node ($n_{\mathit{sel}}$ or $n_{\mathit{fail}}$) matches the ground-truth node determined by manual analysis of $M^{RG}$; (ii) whether every claim in the generated explanation is supported by the backward MEA or diagnosis trace; and (iii) whether defensive outputs are derived strictly from repair predicates (validity) or invalidity causing predicates (invalidity), bounded by the analyst's authorized abstraction view. We present one walkthrough per analysis type; aggregate results are reported in Table~\ref{tab:path-analysis}. Please note that ground-truth explanation nodes and blocking conditions were established by three experts independently analyzing $M^{RG}$ for each query; inter-rater agreement was 100\% for node selection and 98\% for condition identification, with disagreements resolved by discussion.

\smallskip
\noindent\textbf{Path Validity (IT Analyst).}
An IT analyst queries: \textit{``Why does this attack-fault path lead to flare flameout?''}. The proposed path has three steps: exploiting CVE-2018-0296 on the Windows Server, exploiting CVE-2017-9312 on the Allen-Bradley ControlLogix PLC, and triggering flare flameout. This reflects an exploit-centric mental model that captures the direction of attack propagation but overlooks the intermediate access and process conditions required for execution.

AWARE-RG determines the path is repairable (Definition~\ref{def:path-status}) and identifies two conditions absent from the analyst's knowledge: access to PLC port 1132/TCP, required before the PLC can be exploited, and the compromised PLC state that drives the system toward flameout. Defensive strategies are generated strictly from these repair predicates and bounded by the analyst's view: allowlist PLC communication and block unauthorized access to port 1132/TCP; and detect, isolate, and recover the compromised PLC state, preventing downstream fault propagation. Backward MEA identifies $\text{MIT}_0$-$\text{MOT}_0$ as $n_{\mathit{sel}}$, so the analyst receives IT detail with OT consequences at the $C_{11}$ abstraction level only, without the $S_{11}$ fault propagation predicates outside their authorized view. Every defensive strategy is derived from repair evidence rather than from action names or the full validated path. Table~\ref{tab:val_analysis} in Appendix~\ref{app:analysis} reports validity results across queries.

\smallskip
\noindent\textbf{Path Invalidity (OT Analyst).}
An OT analyst asks: \textit{``Why does this attack-fault path not lead to flare flameout?''} The proposed path begins with exploiting the Windows Server and CVE-2017-9312 on the Allen-Bradley PLC and follows a physical fault chain: blocked valve, closed manual isolation valve, low pilot supply pressure, and pilot extinction. This process centric mental model understands downstream OT fault propagation but assumes the PLC exploit is feasible and the controller remains capable of driving the fault sequence.

AWARE-RG begins invalidity analysis at $\text{MOT}_2$. By Proposition~\ref{prop:monotonicity}, the path is valid at $\text{MOT}_2$ but invalid at $\text{MOT}_1$, marking the abstraction boundary at which the analyst's mental model diverges from $M^{RG}$. The diagnosis trace shows the PLC operational state required for the downstream sequence is unsatisfied after the attempted CVE-2017-9312 exploit, failing at the first step; by Lemma~\ref{lem:diagnosis} this precondition is confirmed against $M^{RG}$ as a hardening pointer. Bounded by $\mathit{Pred}(n_{\mathit{fail}})$, the analyst receives recommendations to monitor PLC availability, detect exploit induced controller unavailability, and isolate affected controllers, all within the OT domain. An analyst reasoning only from the proposed path would harden the fault chain; AWARE-RG shows that the defensive leverage also lies at the cyber-physical boundary, where a PLC unavailability condition, the analyst's model assumed away, already blocks propagation. Table~\ref{tab:invalid_analysis} in Appendix~\ref{app:analysis} reports further invalidity results.

\smallskip
\noindent\textbf{Aggregate Results.}
Table~\ref{tab:path-analysis} summarizes all 100 queries. For valid path queries, node selection accuracy is 100\% for IT profiles throughout and degrades on OT queries (92\% repaired, 92\% blocked, 83\% wrong CVE), where limited IT visibility makes abstraction boundary localization harder; repair precision and recall show the same asymmetry ($0.91/0.88$ OT vs.\ $0.95/0.92$ IT). For valid path queries, grounding violations peak at 8\% on OT wrong CVE queries, reflecting natural language variation at the IT/OT abstraction boundary. Four boundary predicate disclosures appear across the 150 invalid-path artifacts. These are not violations of Property~\ref{property:output-confinement}: each lies within $\mathit{Pred}(n_{\mathit{fail}})$ and is surfaced by stepwise diagnosis, and by Lemma~\ref{lem:diagnosis} every hardening pointer is confirmed against $M^{RG}$. No opposite domain $S_i$ predicate is exposed in any case; they are reported for transparency, not as leakage.

\begin{table}[!h]
\centering
\footnotesize
\setlength{\tabcolsep}{3pt}
\renewcommand{\arraystretch}{1.15}
\caption{Path Validity and Invalidity Analysis across 100 queries and IT and OT analyst profiles. \textbf{Node Acc.}: explanation-node selection accuracy against ground truth. \textbf{Repair P/R}: precision and recall of missing condition identification. \textbf{Diag.\ Acc.}: diagnosis accuracy of blocking conditions. \textbf{GVR}: grounding violation rate. \textbf{Leak.}: out-of-view predicate symbol (lexical). \textbf{Bnd.}: boundary predicates disclosed within $\mathit{Pred}(n_{\mathit{fail}})$, authorized rather than leaked. \textbf{--}: not applicable.}
\label{tab:path-analysis}
\resizebox{\columnwidth}{!}{%
\begin{tabular}{l l c c c c c c c c}
\hline
\textbf{Query Variant} & \textbf{Analyst} & \textbf{Queries} &
\textbf{Artifacts} & \textbf{Node Acc.} & \textbf{Repair P/R} &
\textbf{Diag.\ Acc.} & \textbf{GVR} & \textbf{Leak.} & \textbf{Bnd.} \\
\hline
\multicolumn{10}{l}{\textit{Path Validity Analysis}} \\
\hline
Complete & IT & 13 & 39 & 100\% & --        & --    & 0\% & 0/39 & -- \\
Repaired & IT & 12 & 36 & 100\% & 0.95/0.92 & --    & 3\% & 0/36 & -- \\
Complete & OT & 13 & 39 & 100\% & --        & --    & 0\% & 0/39 & -- \\
Repaired & OT & 12 & 36 & 92\%  & 0.91/0.88 & --    & 5\% & 0/36 & -- \\
\hline
\multicolumn{10}{l}{\textit{Path Invalidity Analysis}} \\
\hline
Blocked   & IT & 13 & 39 & 100\% & -- & 100\% & 0\% & 0/39 & 1 \\
Blocked   & OT & 12 & 36 & 92\%  & -- & 92\%  & 4\% & 0/36 & 1 \\
Wrong CVE & IT & 13 & 39 & 100\% & -- & 100\% & 0\% & 0/39 & 1 \\
Wrong CVE & OT & 12 & 36 & 83\%  & -- & 92\%  & 8\% & 0/36 & 1 \\
\hline
\textbf{Total} & \textbf{All} & \textbf{100} & \textbf{300} &
\textbf{96\%} & \textbf{0.93/0.90} & \textbf{96\%} &
\textbf{3\%} & \textbf{0/300} & \textbf{4} \\
\hline
\end{tabular}
}
\end{table}

Table~\ref{tab:ablation} reports the ablation and baseline comparison. The no abstraction baseline exposes all 90 predicates per explanation and leaks in every case, motivating AWARE-RG's abstraction based disclosure control. Vanilla HELM achieves zero leakage but requires $2^{90}$ nodes, making lattice search intractable; FCA and ICS Cyber Kill Chain pruning in AWARE-RG reduces this to 31 nodes with negligible change in mean predicate exposure ($6.1$ vs.\ $6.3$). Disabling bridge nodes leaves 34 of 100 queries without a valid explanation node, confirming that bridge nodes are necessary for cross-domain explanation. These configurations isolate the lattice design's contribution. 
% Ablations of individual pipeline components (Ground, Project, human review) and an LLM-only baseline with prompt based disclosure control are left to future work (\S\ref{future_conclusion}).

% Table~\ref{tab:ablation} reports the ablation and baseline comparison. The no abstraction baseline exposes all 22 predicates per explanation and produces 100\% leakage, motivating AWARE-RG's abstraction based disclosure control. Vanilla HELM achieves zero leakage but requires $2^{22}$ nodes, making lattice search intractable. AWARE-RG reduces the lattice to 31 nodes via FCA and ICS Cyber Kill Chain pruning, with negligible change in mean predicate exposure compared to Vanilla HELM ($6.1$ vs.\ $6.3$ predicates per explanation). Disabling bridge nodes leaves 34 of 100 queries without a valid explanation node, confirming that bridge nodes are necessary for cross-domain explanation.

\begin{table}[!h]
\centering
\footnotesize
\setlength{\tabcolsep}{3pt}
\renewcommand{\arraystretch}{1.15}
\caption{Ablation and baseline comparison. \textbf{Pred.\ Exposed}: mean predicates exposed per explanation. \textbf{Lattice Nodes}: total nodes in the abstraction structure. \textbf{No Expl.\ Node}: queries with no valid explanation node within the analyst's branch. \textbf{Lexical Leakage}: out-of-view predicate leakage.}
\label{tab:ablation}
\resizebox{\columnwidth}{!}{%
\begin{tabular}{l c c c c}
\hline
\textbf{Configuration} &
\textbf{Pred.\ Exposed} &
\textbf{Lattice Nodes} &
\textbf{No Expl.\ Node} &
\textbf{Lexical Leakage Rate} \\
\hline
No Abstraction Baseline ($M^{RG}$) & 90  & --       & 0      & 100\% \\
Vanilla HELM (unpruned)             & 6.3 & $2^{90}$ & 0      & 0\%   \\
No Bridge Ablation                  & 5.1 & 7        & 34/100 & 0\%   \\
\textbf{AWARE-RG (ours)} & \textbf{6.1} & \textbf{31} &
\textbf{0} & \textbf{0\%} \\
\hline
\end{tabular}%
}
\end{table}

%% file: 72-3-ThreatEvaluation.tex
\subsection{Threat Mitigation Analysis}
The three adversarial threat vectors are evaluated across 10 adversarial queries, two per attack strategy, summarized in Table~\ref{tab:threat_eval}. Tables~\ref{tab:illustrative-queries} and~\ref{tab:t3-eval} in Appendix~\ref{app:analysis} give the queries submitted by malicious IT and OT analysts with justifications.

\smallskip
\noindent\textbf{$T_1$: Semantic Mapping Attack.}
We craft an IT existence query asking why an attack-fault path leads to flare flameout, whose second step is deliberately ambiguous at the action boundary: although it describes CVE-2017-9312 behavior on the Allen-Bradley PLC, it is worded as a generic network exploit, creating semantic overlap with \path{exploits-vulnerability-cve-2018-0296}, which targets the Windows Server. Structural verification passes, and the mapper grounds the step to the Windows Server exploit, producing a structurally valid but semantically incorrect plan. Without human-in-the-loop review, the resulting explanation and defensive guidance would be misdirected toward generic vulnerability management rather than targeted PLC port allowlisting and compromised state isolation. The domain expert identifies the incorrect grounding and rejects it before the explanation is generated. $T_1$ is therefore structurally viable, since a semantically incorrect grounding can pass formal checks, but ineffective in practice because semantic correctness is enforced by human review before the explanation stage.

\smallskip
\noindent\textbf{$T_2$: Repair Probing.}
A malicious IT analyst submits deliberately incomplete attack-fault paths, expecting the repair output to disclose the missing precondition chain. AWARE-RG prevents this through the bridge node lattice: explanation search selects an IT-to-OT bridge node at the analyst's authorized abstraction level, with $M^{RG}$ excluded from explanation node selection (Algorithm~\ref{alg:valid-path-mea}), so repair information is limited to authorized boundary predicates such as \path{has-improper-vpn-firewall-configuration}. The detailed OT dependency chain, including device connectivity, PLC vulnerability conditions, and physical fault propagation steps, remains inaccessible. Across all probing attempts, no $S_i$ level OT predicate appeared in any repair or explanation output, so repeated probing cannot reconstruct the complete dependency chain. $T_2$ is mitigated at the symbol level by bridge node selection policy and the minimum disclosure by Property~\ref{prop:ntk}.

% To demonstrate $T_2$, we consider the malicious IT analyst who deliberately submits incomplete attack-fault paths to infer hidden dependencies in the RG planning model from repair outputs. The adversary removes intermediate conditions and expects the repair engine to disclose the missing precondition chain.
% AWARE-RG prevents this reconstruction through the bridge node lattice. During valid path analysis, the explanation search selects an IT-to-OT bridge node at the analyst's authorized abstraction level, while $M^{RG}$ is explicitly excluded from explanation-node selection, as shown in Algorithm \ref{alg:valid-path-mea}. As a result, repair information returned to the IT analyst is limited to authorized boundary predicates, such as \texttt{has-improper-vpn-firewall-configuration} and \texttt{has\_mitre\_ics\_impair\_process\_control}. The detailed OT dependency chain, including device connectivity, PLC vulnerability conditions, and physical fault propagation steps, remains inaccessible.

% Across all evaluated probing attempts, no $S_i$ level OT predicate appeared in any repair or explanation output returned to the adversary. Thus, an adversary probing from the IT branch receives only high-level cross-domain context consistent with the bridge node policy, not the fine-grained structure of the complete RG model. This prevents repeated repair probing from reconstructing the complete dependency chain. Hence, $T_2$ is mitigated by the bridge node selection policy and the minimum disclosure guarantee in Property \ref{prop:ntk}.
\smallskip
\noindent\textbf{$T_3$: Explanation Inference.}
We construct a query pair in which an adversary submits the same query, asking why the path leads to flameout, as both an IT and an OT analyst, then compares outputs to infer predicate level abstraction boundaries in the opposite domain. The IT explanation, grounded in the backward MEA trace at the IT-to-OT bridge node, exposes IT predicates such as \path{has-access-to-plc-port-1132-tcp} with OT context only at $C_{11}$; no $S_{11}$ fault propagation predicates appear. The OT explanation, grounded at the OT-to-IT bridge node, exposes OT process fault predicates such as \path{has-fault-valve-blocked-close-due-to-compromised} with IT context only at $C_1$; no IT side $S_i$ predicates appear. Since the lattice structure, node selection policy, and predicate mappings remain internal to AWARE-RG, the adversary cannot determine whether withheld predicates exist or where layer visibility boundaries lie. By Proposition~\ref{prop:disclosure-bound}, $T_3$ is bounded: the adversary learns the authorized $C_i$ skeleton, which is permitted by disclosure policy.

\begin{table}[!h]
\centering
\scriptsize
\setlength{\tabcolsep}{3pt}
\renewcommand{\arraystretch}{1.15}
\caption{Threat mitigation evaluation. \textbf{Info Leaked}: whether any opposite domain $S_i$ predicate appeared in an output returned to the adversary}
\label{tab:threat_eval}
\resizebox{\columnwidth}{!}{%
\begin{tabular}{p{0.4cm} p{3.0cm} p{2.2cm} p{0.9cm} p{0.9cm}}
\hline
\textbf{T} &
\textbf{Attack Strategy} &
\textbf{Mitigation} &
\textbf{Effective} &
\textbf{Info Leaked} \\
\hline
$T_1$ & Ambiguous CVE description & Human review & Yes & No \\
% \hline
$T_1$ & Generic action phrasing & Human review & Yes & No \\
% \hline
$T_2$ & Incomplete IT path probing & Bridge node policy & Yes & No \\
% \hline
$T_2$ & Incomplete OT path probing & Bridge node policy & Yes & No \\
% \hline
$T_3$ & Same query, IT vs.\ OT output comparison & Node selection policy & Yes & No \\
\hline
% $T_3$ & Same query, IT vs.\ OT output comparison & Node selection policy & Yes & No \\
% \hline
\end{tabular}%
}
\end{table}

%% file: 90-FutureDirectionandConclusion.tex
\section{Conclusion and Future Work} \label{conclusion}
We presented AWARE-RG, a semi-automated framework that explains CPS attack-fault paths to IT/OT analysts on a \textit{need-to-know} basis aligned with their expertise. By explaining at the analyst's knowledge level, it delivers the minimum information needed without overwhelming them or exceeding their \textit{need-to-know} principle. To the best of our knowledge, AWARE-RG is the first framework to address this problem through analyst aware abstraction guided explanation. It formally confines returned predicates to the analyst's authorized abstraction and limits LLM hallucination through expert validated mapping and evidence grounding. Although evaluated on a single CPS using RG, it applies in principle to any planning compilation of CPS model. In the future, we plan to measure semantic leakage in AWARE-RG's explanations through independent annotation. We will also evaluate the framework across LLMs of different scales, compare it against an LLM-only baseline, and run ablations of its individual components. To test generality, we intend to validate AWARE-RG on a different CPS using held-out queries and external annotators. Finally, we plan to add stateful explanation histories that preserve need-to-know boundaries and to conduct studies with ICS practitioners.

%% file: Appendix_New.tex
\appendices
\onecolumn
\section{Path Validity \& Invalidity, and Threat Analysis}
\label{app:analysis}

\begin{table*}[!htbp]
\centering
\scriptsize
\setlength{\tabcolsep}{4pt}
\renewcommand{\arraystretch}{1.3}
\caption{Path Invalidity analysis results across IT and OT analyst profiles. Each row shows a query submitted to AWARE-RG, the hardening recommendation derived strictly from the invalidity causing predicates at $\mathit{Pred}(n_{\mathit{fail}})$ and runtime.}
\label{tab:invalid_analysis}
\resizebox{\textwidth}{!}{%
\begin{tabularx}{\textwidth}{
>{\RaggedRight\arraybackslash}p{1.2cm}|
>{\RaggedRight\arraybackslash}X|
>{\RaggedRight\arraybackslash}p{4.8cm}|
>{\centering\arraybackslash}p{1.2cm}
}
\toprule
\textbf{Security Analyst} &
\textbf{Query \& Steps} &
\textbf{Hardening Recommended by AWARE-RG} &
\textbf{Runtime (s)} \\
\midrule

IT &
\textbf{Query:} Why does this attack-fault path not lead to flare
flameout?
\newline
\textbf{Steps:}
\newline
1. Exploit CVE-2018-0296 on Windows Server to gain internal network
access.
\newline
2. Exploit CVE-2017-9312 on Allen-Bradley ControlLogix PLC.
\newline
3. PLC compromise causes flare flameout. &
Restrict or monitor access to Allen-Bradley ControlLogix PLC port
1132/TCP to block the attack path. &
156 \\
\midrule

IT &
\textbf{Query:} Why does this attack-fault path not lead to flare
flameout?
\newline
\textbf{Steps:}
\newline
1. Exploit CVE-2017-9312 on Allen-Bradley ControlLogix PLC.
\newline
2. Compromised PLC causes process valve to become blocked closed.
\newline
3. Blocked valve causes manual isolation valve to close.
\newline
4. Closed isolation valve causes low pilot supply pressure.
\newline
5. Low pilot supply pressure leads to pilot extinction.
\newline
6. Pilot extinction causes flare flameout. &
The blocking condition is the absence of a blocked process valve,
which is a prerequisite for the manual isolation valve to close.
Preserve this condition to maintain the resilience boundary. &
134 \\
\midrule

IT &
\textbf{Query:} Why does this attack-fault path not lead to flare
flameout?
\newline
\textbf{Steps:}
\newline
1. Compromised PLC causes process valve to become blocked closed.
\newline
2. Blocked valve causes manual isolation valve to close.
\newline
3. Closed isolation valve causes low pilot supply pressure.
\newline
4. Low pilot supply pressure leads to pilot extinction.
\newline
5. Pilot extinction causes flare flameout. &
Monitor and detect PLC process control integrity; promptly identify
and address any impairment to preserve the blocking condition. &
104 \\
\midrule

OT &
\textbf{Query:} Why does this attack-fault path not lead to flare
flameout?
\newline
\textbf{Steps:}
\newline
1. Closed isolation valve causes low pilot supply pressure.
\newline
2. Low pilot supply pressure leads to pilot extinction.
\newline
3. Pilot extinction causes flare flameout. &
The blocking condition is the open isolation valve; ensure the
isolation valve state is monitored and preserved to prevent
downstream fault propagation. &
75 \\
\midrule

OT &
\textbf{Query:} Why does this attack-fault path not lead to flare
flameout?
\newline
\textbf{Steps:}
\newline
1. Exploit CVE-2018-0296 on Windows Server to gain internal network
access.
\newline
2. Access PLC port 1132/TCP on Allen-Bradley ControlLogix PLC.
\newline
3. Exploit CVE-2017-9312 on Allen-Bradley ControlLogix PLC.
\newline
4. Compromised PLC causes process valve to become blocked closed.
\newline
5. Blocked valve causes manual isolation valve to close.
\newline
6. Closed isolation valve causes low pilot supply pressure.
\newline
7. Low pilot supply pressure leads to pilot extinction.
\newline
8. Pilot extinction causes flare flameout. &
Path is not invalid. &
151 \\
\midrule

OT &
\textbf{Query:} Why does this attack-fault path not lead to flare
flameout?
\newline
\textbf{Steps:}
\newline
1. Exploit CVE-2018-0296 on Windows Server to gain internal network
access.
\newline
2. Exploit CVE-2017-9312 on Allen-Bradley ControlLogix PLC.
\newline
3. Compromised PLC causes process valve to become blocked closed.
\newline
4. Blocked valve causes manual isolation valve to close.
\newline
5. Pilot extinction causes flare flameout. &
Restrict access to PLC port 1132/TCP to prevent exploitation of
CVE-2017-9312 and block the attack-fault path. &
141 \\
\bottomrule
\end{tabularx}%
}
\end{table*}

\begin{table*}[!htbp]
\centering
\scriptsize
\setlength{\tabcolsep}{4pt}
\renewcommand{\arraystretch}{1.3}
\caption{Path Validity Analysis results across IT and OT analyst profiles. \textbf{Repair Count} denotes the number of actions inserted by plan repair, representing conditions absent from the analyst's knowledge. \textbf{Defensive Recommendations} are derived strictly from repair predicates and bounded by the analyst's authorized abstraction view. \textbf{Runtime} reports end-to-end execution time in seconds.}
\label{tab:val_analysis}
\resizebox{\textwidth}{!}{%
\begin{tabularx}{\textwidth}{
>{\RaggedRight\arraybackslash}p{1.2cm}|
>{\RaggedRight\arraybackslash}X|
>{\centering\arraybackslash}p{1.3cm}|
>{\RaggedRight\arraybackslash}X|
>{\centering\arraybackslash}p{1.2cm}
}
\toprule
\textbf{Security Analyst} &
\textbf{Query \& Steps} &
\textbf{Repair Count} &
\textbf{Defensive Recommendations by AWARE-RG} &
\textbf{Execution (s)} \\
\midrule

IT &
\textbf{Query:} Why does this attack-fault path lead to flare
flameout?
\newline
\textbf{Steps:}
\newline
1. Access PLC port 1132/TCP on Allen-Bradley ControlLogix PLC from
internal network.
\newline
2. Exploit CVE-2017-9312 on Allen-Bradley ControlLogix PLC.
\newline
3. Compromised PLC causes process valve to become blocked closed. &
2 &
Ensure proper firewall configuration to prevent unauthorized access
to the internal network. &
177 \\
\midrule

IT &
\textbf{Query:} Why does this attack-fault path lead to flare
flameout?
\newline
\textbf{Steps:}
\newline
1. Access PLC port 1132/TCP on Allen-Bradley ControlLogix PLC from
internal network.
\newline
2. Compromised PLC causes process valve to become blocked closed. &
1 &
1. Ensure proper firewall configuration to prevent unauthorized
internal network access.
\newline
2. Implement robust process-control measures to prevent impairment. &
141 \\
\midrule

IT &
\textbf{Query:} Why does this attack-fault path lead to flare
flameout?
\newline
\textbf{Steps:}
\newline
1. Exploit CVE-2017-9312 on Allen-Bradley ControlLogix PLC.
\newline
2. Compromised PLC causes process valve to become blocked closed. &
1 &
1. Segment network and require strong authentication before the
missing network-access condition can become true.
\newline
2. Allowlist PLC communication and block unauthorized access to PLC
port.
\newline
3. Detect, isolate, and recover compromised asset state.
\newline
4. Add monitoring, alarms, or interlocks around the fault. &
122 \\
\midrule

OT &
\textbf{Query:} Why does this attack-fault path lead to flare
flameout?
\newline
\textbf{Steps:}
\newline
1. Compromised PLC causes process valve to become blocked closed.
\newline
2. Blocked valve causes manual isolation valve to close.
\newline
3. Closed isolation valve causes low pilot supply pressure.
\newline
4. Low pilot supply pressure leads to pilot extinction.
\newline
5. Pilot extinction causes flare flameout. &
4 &
Ensure PLC remains online and operational to prevent unauthorized
access and denial-of-service attacks. &
168 \\
\midrule

OT &
\textbf{Query:} Why does this attack-fault path lead to flare
flameout?
\newline
\textbf{Steps:}
\newline
1. Compromised PLC causes process valve to become blocked closed.
\newline
2. Blocked valve causes manual isolation valve to close.
\newline
3. Closed isolation valve causes low pilot supply pressure.
\newline
4. Pilot extinction causes flare flameout. &
3 &
Monitor for missing process-control impairment condition and trigger
isolation or safe-mode response. &
182 \\
\midrule

OT &
\textbf{Query:} Why does this attack-fault path lead to flare
flameout?
\newline
\textbf{Steps:}
\newline
1. Compromised PLC causes process valve to become blocked closed.
\newline
2. Closed isolation valve causes low pilot supply pressure.
\newline
3. Low pilot supply pressure leads to pilot extinction.
\newline
4. Pilot extinction causes flare flameout. &
3 &
Ensure PLC remains online and operational to prevent unauthorized
access and denial-of-service attacks. &
175 \\
\bottomrule
\end{tabularx}%
}
\end{table*}

\begin{table*}[!htbp]
\centering
\scriptsize
\setlength{\tabcolsep}{4pt}
\renewcommand{\arraystretch}{1.3}
\caption{Illustrative queries submitted to AWARE-RG across IT and OT analyst profiles. \textbf{Lattice Node} denotes the explanation node selected by AWARE-RG. \textbf{Attacker's Intention} describes the adversarial goal inferred from the query. \textbf{AWARE-RG Explanation} summarizes the generated \textit{need-to-know} explanation. \textbf{Justification} describes how the attack is bounded or stopped.}
\label{tab:illustrative-queries}
\resizebox{\textwidth}{!}{%
\begin{tabular}{p{0.8cm} p{5.2cm} p{2.0cm} p{2.2cm} p{1.6cm} p{3.2cm}}
\hline
\textbf{Analyst (Malicious)} &
\textbf{Query \& Steps} &
\textbf{Attacker's Intention} &
\textbf{AWARE-RG Explanation} &
\textbf{Lattice Node} &
\textbf{Justification} \\
\hline

% IT &
% \raggedright
% \textit{Why does this attack-fault path not lead to flare flameout?}
% \newline
% 1. Exploit CVE-2018-0296 on Windows Server to gain internal network access.
% \newline
% 2. Access PLC port 1132/TCP on Allen-Bradley ControlLogix PLC.
% \newline
% 3. Exploit CVE-2017-9312 on Allen-Bradley ControlLogix PLC.
% \newline
% 4. Compromised PLC causes valve to become blocked closed.
% \newline
% 5. Blocked valve causes manual isolation valve to close.
% \newline
% 6. Closed isolation valve causes low pilot supply pressure.
% \newline
% 7. Low pilot supply pressure leads to pilot extinction.
% \newline
% 8. Pilot extinction causes flare flameout. &
% Propagate a full IT-to-OT attack-fault chain to trigger flare flameout via PLC compromise. &
% Path is blocked: PLC port 1132/TCP access is absent from the initial state; the exploit prerequisite is unsatisfied before step 3 executes. Hardening pointer: preserve port access restrictions. &
% $\text{MIT}_3$-$\text{MOT}_2$ &
% The path fails at step 2 because PLC port 1132/TCP is unreachable; Property~\ref{prop:ntk} confines the explanation to $S_1$ and $C_{11}$; no $S_{11}$ OT predicates are exposed. \\
% \hline

IT &
\raggedright
\textit{Can internal network access allow the attacker to reach PLC port 1132/TCP?}
\newline
1. Exploit VPN/firewall vulnerability on Windows Server.
\newline
2. Access PLC port 1132/TCP from internal network. &
Determine whether firewall exploitation provides a viable path to the PLC communication port. &
Path is repairable: firewall exploitation establishes internal network access ($S_1$, $S_2$), but direct PLC port reachability requires additional network conditions absent from the analyst's model. Missing knowledge: intermediate network segmentation predicates. &
$\text{MIT}_0$-$\text{MOT}_2$ &
Repair identifies missing network predicates; defensive strategy is bounded to IT abstraction view with OT consequence communicated at $C_{11}$ only. \\
\hline

OT &
\raggedright
\textit{Why does this attack-fault path lead to flare flameout?}
\newline
1. Exploit CVE-2017-9312 from internal network to compromise Allen-Bradley ControlLogix PLC.
\newline
2. Compromised PLC causes valve to become blocked closed.
\newline
3. Blocked valve causes manual isolation valve to close.
\newline
4. Closed isolation valve causes low pilot supply pressure.
\newline
5. Low pilot supply pressure causes pilot extinction.
\newline
6. Pilot extinction causes flare flameout. &
Confirm that PLC compromise propagates through the physical fault chain to trigger flare flameout. &
Path is repairable: OT fault chain is valid but IT side access prerequisites (network access, port reachability) are absent from the analyst's model. Missing knowledge: IT access conditions communicated at $C_1$ abstraction level. &
$\text{MOT}_2$-$\text{MIT}_3$ &
Repair exposes missing IT prerequisites at $C_1$ only; no IT side $S_i$ predicates are disclosed to the OT analyst, consistent with Property~\ref{prop:ntk}. \\
\hline
\end{tabular}%
}
\end{table*}

\clearpage

% \begin{table}[!htbp]
% \centering
% \scriptsize
% \setlength{\tabcolsep}{3pt}
% \renewcommand{\arraystretch}{1.3}
% \caption{$T_3$ evaluation: same query submitted by IT and OT analysts to demonstrate that AWARE-RG bounds explanation inference. Each query pair shares a single attacker intention and justification.%; analyst specific columns show how outputs differ by role without exposing opposite domain $S_i$ predicates.
% }
% \label{tab:t3-eval}
% \resizebox{\textwidth}{!}{%
% \begin{tabular}{p{0.6cm} p{1.0cm} p{5.0cm} p{2.2cm} p{1.4cm} p{3.0cm} p{2.8cm}}
% \hline
% \textbf{Query ID} &
% \textbf{Analyst Role} &
% \textbf{Query \& Steps} &
% \textbf{Attacker's Intention} &
% \textbf{Lattice Node} &
% \textbf{AWARE-RG Explanation} &
% \textbf{Justification} \\
% \hline
\par\medskip
\noindent\begin{minipage}{\textwidth}
\centering
\scriptsize
\setlength{\tabcolsep}{3pt}
\renewcommand{\arraystretch}{1.3}
\captionof{table}{$T_3$ evaluation: same query submitted by IT and OT analysts to demonstrate that AWARE-RG bounds explanation inference. Each query pair shares a single attacker intention and justification; analyst specific columns show how outputs differ by role without exposing opposite domain $S_i$ predicates.}
\label{tab:t3-eval}

\resizebox{\textwidth}{!}{%
\begin{tabular}{p{0.6cm} p{1.0cm} p{5.0cm} p{2.2cm} p{1.4cm} p{3.0cm} p{2.8cm}}
\hline
\textbf{Query ID} &
\textbf{Analyst (Malicious)} &
\textbf{Query \& Steps} &
\textbf{Attacker's Intention} &
\textbf{Lattice Node} &
\textbf{AWARE-RG Explanation} &
\textbf{Justification} \\
\hline

\multirow{2}{*}{Q1} &
IT (Valid) &
\multirow{2}{5.0cm}{%
\raggedright
\textit{Why does this attack-fault path lead to flare flameout?}
\newline
1. Exploit CVE-2018-0296 on Windows Server to gain internal network access.
\newline
2. Exploit CVE-2017-9312 on Allen-Bradley ControlLogix PLC.
\newline
3. PLC compromise causes flare flameout.} &
\multirow{2}{2.2cm}{%
\raggedright
Determine the bridging condition between IT exploitation and OT physical consequence.} &
$\text{MIT}_0$-$\text{MOT}_0$ &
Path is repairable. IT analyst receives $S_1$, $S_2$, $S_6$, $S_7$ detail with OT cross-domain context at $C_{11}$ only. Missing knowledge: PLC port access and compromised PLC state. &
\multirow{2}{2.8cm}{%
\raggedright
Comparing IT and OT outputs, the adversary observes only $C_{11}$ (IT side) and $C_1$ (OT side). No $S_i$ opposite domain predicate is exposed in either output. Proposition~\ref{prop:disclosure-bound} bounds inference: the adversary learns at most the authorized $C_i$ skeleton, not the hidden $S_i$ structure.} \\
\cline{2-2}\cline{5-6}

& OT (Valid) &
&
&
$\text{MOT}_0$-$\text{MIT}_0$ &
Path is repairable. OT analyst receives $S_2$, $S_7$, $S_{11}$, $S_{12}$ detail with IT cross-domain context at $C_1$ only. Missing knowledge: IT access prerequisites. &
\\
\hline

\multirow{2}{*}{Q2} &
IT (Invalid) &
\multirow{2}{5.0cm}{%
\raggedright
\textit{Why does this attack-fault path not lead to flare flameout?}
\newline
1. Exploit CVE-2017-9312 on Allen-Bradley ControlLogix PLC.
\newline
2. Compromised PLC causes valve to become blocked closed.
\newline
3. Blocked valve causes manual isolation valve to close.
\newline
4. Closed isolation valve causes low pilot supply pressure.
\newline
5. Low pilot supply pressure causes pilot extinction.
\newline
6. Pilot extinction causes flare flameout.} &
\multirow{2}{2.2cm}{%
\raggedright
Determine the bridging condition that blocks the OT fault chain from the IT side.} &
$\text{MIT}_3$ &
Path is blocked at $\text{MIT}_3$. IT analyst receives hardening pointer: PLC port 1132/TCP access is absent; no OT predicates exposed. &
\multirow{2}{2.8cm}{%
\raggedright
IT output exposes only $S_1$ and blocking condition at the IT abstraction boundary. OT output exposes $S_{7}$, $S_{11}$, $S_{12}$ fault predicates with IT context at $C_1$ only. Adversary cannot reconstruct opposite domain visibility from either output, consistent with Proposition~\ref{prop:disclosure-bound}.} \\
\cline{2-2}\cline{5-6}

& OT (Invalid) &
&
&
$\text{MOT}_2$ &
Path is blocked at $\text{MOT}_2$. OT analyst receives hardening pointer: PLC operational state unsatisfied; IT prerequisites communicated at $C_1$ only. &
\\
& & & & & & \\
% \hline
% \end{tabular}%
% }
% \end{table}

\hline
\end{tabular}%
}
\end{minipage}
\par\vspace{1.5em}

% \FloatBarrier
% \vspace{-0.5em}

\begin{multicols}{2}

\section{The Flare System}
\label{app:flare}

\begin{center}
    \includegraphics[width=0.95\linewidth]{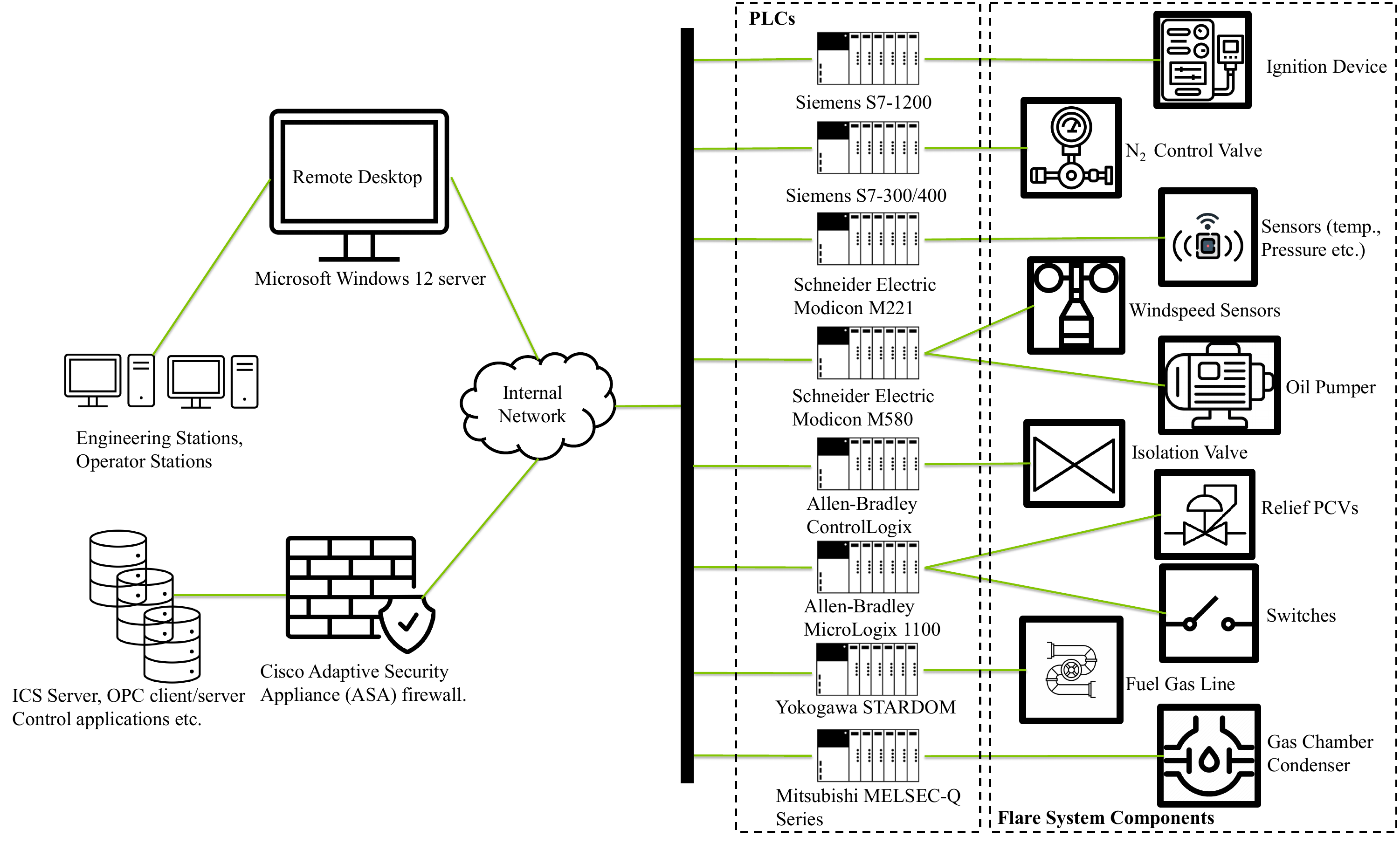}
    \captionof{figure}{Schematic overview of the Flare System's IT and OT components.}
    \label{fig:flare}
\end{center}

\begin{center}
\tiny
\setlength{\tabcolsep}{2pt}
\renewcommand{\arraystretch}{1.15}
\captionof{table}{Flare System PLCs, CVE vulnerabilities, and functionality.}
\label{tab:plc_cves}

\begin{tabularx}{\linewidth}{p{2.5cm}p{1.55cm}X}
% \rowcolor{lightgray}
\hline
\textbf{PLC Name} & \textbf{CVE(s)} & \textbf{Operation in OT Network} \\
\hline
Siemens S7-1200 & CVE-2017-12741, CVE-2018-13810 &
Control the Ignition Device. \\
% \hline
Siemens S7-300/400 & CVE-2017-12741, CVE-2016-8673 &
Control opening/closing of Nitrogen Valve. \\
% \hline
Schneider Modicon M221 & CVE-2017-6032, CVE-2017-6034 &
Monitor Mechanical \& Instrumental sensor readings. \\
% \hline
Schneider Modicon M580 & CVE-2019-6815, CVE-2018-7789 &
Monitor windspeed and pumping of flare system. \\
% \hline
Allen-Bradley ControlLogix & CVE-2017-9312 &
Control the isolation valve. \\
% \hline
Allen-Bradley MicroLogix 1100 & CVE-2017-12089 &
Control relief PCVs and flare switching. \\
% \hline
Yokogawa STARDOM & CVE-2018-10594 &
Control and monitor Fuel Gas line. \\
% \hline
Mitsubishi MELSEC-Q & CVE-2017-9638 &
Monitor condensate in Fuel Gas. \\
\hline
\end{tabularx}
\end{center}
% \end{multicols}

% \begin{multicols}{2}
\section{High to low-level Predicate Mapping}
\label{app:evaluation}

\noindent
\begin{minipage}{\linewidth}
\centering
\tiny
\setlength{\tabcolsep}{2pt}
\renewcommand{\arraystretch}{1.15}
\captionof{table}{Mapping between low-level ICS Cyber Kill Chain stages and high-level abstraction predicates used in bridge node construction.}
\label{tab:low_high_bridge_mapping}
\vspace{2mm}
\resizebox{\linewidth}{!}{%
\begin{tabular}{c|p{3cm} c c p{4.2cm}}
\hline
\textbf{Domain} &
\textbf{MITRE ATT\&CK for ICS Tactic} &
\textbf{Low-Level} &
\textbf{High-Level} &
\textbf{High-Level Predicate Representation} \\
\hline

IT    & Initial Access              & $S_1$    & $C_1$    & \texttt{has\_initial\_access} \\
IT/OT & Execution                   & $S_2$    & $C_2$    & \texttt{has\_execution} \\
IT    & Persistence                 & $S_3$    & $C_3$    & \texttt{has\_persistence} \\
IT    & Privilege Escalation        & $S_4$    & $C_4$    & \texttt{has\_privilege\_escalation} \\
IT    & Command and Control         & $S_5$    & $C_5$    & \texttt{has\_command\_and\_control} \\
IT    & Discovery                   & $S_6$    & $C_6$    & \texttt{has\_discovery} \\
IT/OT & Lateral Movement            & $S_7$    & $C_7$    & \texttt{has\_lateral\_movement} \\
IT    & Collection                  & $S_8$    & $C_8$    & \texttt{has\_collection} \\
IT    & Evasion                     & $S_9$    & $C_9$    & \texttt{has\_evasion} \\
OT    & Inhibit Response Function   & $S_{10}$ & $C_{10}$ &
\path{has_inhibit_response_function} \\
OT    & Impair Process Control      & $S_{11}$ & $C_{11}$ &
\texttt{has\_impair\_process\_control} \\
OT    & Impact                      & $S_{12}$ & $C_{12}$ &
\texttt{has\_impact} \\
\hline
\end{tabular}%
}
\end{minipage}
% \end{multicols}

% \begin{multicols}{2}
\section{Property 2 Verification Results}
\label{app:property2}

\noindent
\begin{minipage}{\linewidth}
\centering
\tiny
\setlength{\tabcolsep}{2pt}
\renewcommand{\arraystretch}{1.15}
\captionof{table}{
Property~\ref{property:output-confinement} verification across
100 queries and 300 artifacts. An out-of-view token is an
$S_i$ predicate outside the analyst's authorized explanation node.
}
\label{tab:property2}
\vspace{2mm}
\resizebox{\linewidth}{!}{%
\begin{tabular}{l l c c c c}
\hline
\textbf{Query Type} &
\textbf{Analyst} &
\textbf{Queries} &
\textbf{Artifacts} &
\textbf{Out-of-View Tokens} &
\textbf{Leakage Rate} \\
\hline

Path Validity (Complete)    & IT & 13 & 39 & 0 & 0/39 \\
Path Validity (Repaired)    & IT & 12 & 36 & 0 & 0/36 \\
Path Validity (Complete)    & OT & 13 & 39 & 0 & 0/39 \\
Path Validity (Repaired)    & OT & 12 & 36 & 0 & 0/36 \\
Path Invalidity (Blocked)   & IT & 13 & 39 & 0 & 0/39 \\
Path Invalidity (Blocked)   & OT & 12 & 36 & 0 & 0/36 \\
Path Invalidity (Wrong CVE) & IT & 13 & 39 & 0 & 0/39 \\
Path Invalidity (Wrong CVE) & OT & 12 & 36 & 0 & 0/36 \\
\hline

\textbf{Total} & \textbf{All} &
\textbf{100} & \textbf{300} &
\textbf{0} & \textbf{0/300} \\
\hline
\end{tabular}%
}
\end{minipage}

\end{multicols}